\documentclass[a4paper]{llncs}

\usepackage{times}
\usepackage{amssymb}
\usepackage{amsmath}
\usepackage{graphicx}
\usepackage{comment}
\usepackage{stmaryrd}
\usepackage{mathpartir}
\usepackage[percent]{overpic}
\usepackage[ruled]{algorithm}
\usepackage{algpseudocode}
\usepackage{multicol}
\usepackage{verbatim}
\usepackage{sidecap}

\newcommand{\texta}[1]{{\small\textsf{#1}}}
\renewcommand{\tilde}[1]{\widetilde{#1}}
\newcommand{\loc}{\text{\emph{loc}}}
\renewcommand{\read}{\text{\emph{read}}}

\newtheorem{notation}{Notation}

\begin{document}

\title{Verification of agent knowledge in\\ dynamic access control policies$^*$}

\author{Masoud Koleini, Eike Ritter \and Mark Ryan \\ 
}
%

\institute{School of Computer Science \\ The University of Birmingham \\ Birmingham, B15 2TT, UK}

\maketitle

\begin{abstract}

\let\thefootnote\relax\footnotetext{* The original version of this paper appeared in Lecture Notes in Computer Science (LNCS), Volume 7795, 2013, pp 448--462.}

We develop a modeling technique based on interpreted systems in order to verify temporal-epistemic properties over access control policies. This approach enables us to detect information flow vulnerabilities in dynamic policies by verifying the knowledge of the agents gained by both reading and reasoning about system information. To overcome the practical limitations of state explosion in model-checking temporal-epistemic properties, we introduce a novel abstraction and refinement technique for temporal-epistemic safety properties in ACTLK (ACTL with knowledge modality K) and a class of interesting properties that does fall in this category. 

\end{abstract}




\section{Introduction}

Assume a conference paper review system in which all the PC members have access to the number of the papers assigned to each reviewer. Further assume that a PC member Alice can see the list of the papers that are assigned to another PC member and that are not authored by Alice. Then if Alice is the author of a submitted paper, she can find who the reviewer of her paper is by comparing the number of papers assigned to each reviewer (shown by the system) with the number of the assigned papers of that reviewer which she has access to.

The above is an example of a potential information leakage in \emph{content management systems}, which are collaborative environments that allow users to create, store and manage data. They also allow controlling access to the data based on the user roles. In such multi-agent systems, access to the data is regulated by \emph{dynamic access control policies}, which are a class of authorization rules that the permissions for an agent depend on the state of the system and change when agents interact with the system \cite{DynPAL,synthesising-nmd-2007,DynACS-Dougherty}. In complicated access control scenarios, there is always a risk that some required properties do not hold in the system. For instance and for a conference paper review system, the following properties need to hold in the policy:

\begin{itemize}
	\item It should be impossible for the author of a paper to be assigned as the reviewer of his own paper (temporal safety property).
	\item There must be no way for the author of a paper to find out who is the reviewer of his paper (epistemic safety property).
\end{itemize}

Epistemic properties take \emph{knowledge} of the agents into account. The knowledge can be gained by directly accessing the information, which complies with one of the meanings of the knowledge in ordinary language, that means the agent \emph{sees} the truth. But agent also knows the truth when he indirectly reasons about it \cite{Mardare:2006:DynamicESLogics}.

Information flow as a result of reasoning is a critical vulnerability in many collaborative systems like conference paper review systems, social networks and document management systems, and is difficult to detect. The complication of access control policies in multi-agent collaborative frameworks makes finding such weaknesses more difficult using non-automated mechanisms. Moreover, the state of art dynamic access control verification tools are unable to find such properties as they do not handle epistemic property verification in general. Therefore as the \emph{first contribution} of this paper, we propose a policy authorization language and express how to use the \emph{interpreted systems framework} \cite{Fagin-book-95} in order to model the related access control system. Using interpreted systems enables us to address misconfiguration in the policy and information disclosure to unauthorized agents by verifying temporal-epistemic properties expressed in the logic CTLK (CTL with knowledge modality K). The knowledge of an agent in our modelling covers both the knowledge gained by reasoning and by reading information when access permission is granted.

The practical limitation of interpreted systems is the state explosion for the systems of medium to large state space. There is also a limited number of research on the automated abstraction and refinement of the models defined in interpreted systems framework. As the \emph{second contribution}, we develop an novel fully automated abstraction and refinement technique for verifying safety properties in ACTLK (which is a subset of CTLK) over an access control system modelled in the framework of interpreted systems. We extend counterexample guided abstraction refinement \cite{CEGAR} to cover the counterexamples generated by the verification of temporal-epistemic properties and when the counterexample is tree-like \cite{Clarke-treelikecounterexamples}. In this paper, we only discuss the counterexamples with finite length paths, but this approach can be extended to the paths of infinite length using an unfolding mechanism \cite{CEGAR}. We use a model-checker for multi-agent systems \cite{Lomuscio:2006:MCMAS} and build the abstract model in its modelling language. The refinement is guided using the counterexample generated by the model-checker. The counterexample checking algorithm is provably sound and complete. We also introduce an interactive refinement for a class of epistemic properties that does not fall in ACTLK, but can specify interesting security properties.

The reminder of the technical report is organized as follows: Related works are discussed in section \ref{sec:relatedw}, interpreted systems are introduced in section \ref{sec:is}, formal syntax and semantics of access control policies are provided in section \ref{sec:syntax}, deriving an interpreted system from a policy is described in section \ref{sec:buildis}, abstraction and refinement technique is given in sections \ref{sec:absref} and \ref{sec:ref}. Case studies and experimental results are included in section \ref{sec:experiment}.


\section{Related work}
\label{sec:relatedw}

In the area of knowledge-based policy verification, Aucher et al. \cite{Aucher2010-modal} define \emph{privacy policies} in terms of permitted or forbidden knowledge. The dynamic part of their logic deals with sending or broadcasting data. Their approach is limited in modeling knowledge gained by the interaction of the agents in a multi-agent system. RW framework \cite{synthesising-nmd-2007} has the most similar approach with ours. The transition system in RW is build over the knowledge of the active coalition of agents. In each state, the knowledge of the coalition is the accumulation of the knowledge obtained by performing actions or sampling system variables in previous transitions together with the initial knowledge. In the other words, knowledge in RW is gained by reading or altering system variables, not by reasoning about them. This is similar to PoliVer \cite{Koleini:2011:PoliVer}, which approximates knowledge by readability. Such verification tools are not able to detect information flow as a result of reasoning.

In the field of abstraction and refinement for temporal-epistemic logic, Cohen et al. \cite{Abstract-MAS-Cohen} introduce the theory of simulation relation and existential abstraction for interpreted systems. Their approach is not automated and they have not provided how to refine the abstract model if the property does not hold and the counterexample is spurious. A recent research on abstraction and refinement for interpreted systems is done by Zhou et al. \cite{Zhou-CEGAR}. Although their work is about abstraction and refinement of interpreted systems, their paper is abstract and mainly discusses the technique to build up a tree-like counterexample when verifying ACTLK properties.


\section{Background}
\label{sec:is}

\subsection{Interpreted systems}
Fagin et al. \cite{Fagin:97:knowledge} introduced interpreted systems as the framework to model multi-agent systems in games scenarios. They introduced a detailed transition system which contains agents, local states and actions. Such a framework enables reasoning about both temporal and epistemic properties of the system. Lomuscio et al \cite{Lomuscio-mck-06} have used a variant of interpreted systems to verify ATLK (alternating time temporal logic \cite{ATL-Alure} with knowledge) properties over the interpreted systems. They have also developed a model-checker for interpreted systems called MCMAS \cite{Lomuscio:2006:MCMAS} which we will use as the model-checking engine in our implementation.

The multi-agent system formalism known as \emph{interpreted systems (IS)} \cite{Fagin-book-95,Fagin:97:knowledge} contains a set $\Omega=\{e,1,\dots ,n\}$ of agents including the \emph{environment} $e$ with the same specification as the other agents. Interpreted systems contain the following elements:

\begin{itemize}
 \item \textbf{Local states:} Each agent in a multi-agent framework has its own local state. The set of local states for the agent $i$ is denoted by $L_i$. The local state of an agent represents the information the agent has direct access to. The environment can be seen as the agent which is capable of capturing or holding the information that is inaccessible to the other agents. For example, the communication channel in a bit transmission protocol can be modelled as the environment. The set of \emph{global states} is $S=L_e\times L_1\times\dots\times L_n$, representing the system at a specific time. The system evolves as a function over the time. We also use the notation of $L_i$ as the function that accepts a set of global states and returns the corresponding set of local states for agent $i$. For each $s\in S$, $l_i(s)$ denotes the local state of agent $i$ in $s$.

\item \textbf{Actions:} State transitions are the result of performing actions by different agents. If $i\in\Omega$, then $ACT_i$ is the set of actions accessible for the agent $i$. The set of \emph{joint actions} is defined as $ACT=ACT_e\times ACT_1\times\dots\times ACT_n$. We also use $ACT_i$ as the function that accepts a joint action and returns the action of agent $i$.

\item \textbf{Protocols:} Protocols are defined as mappings from the set of local states to the set of local actions and define the actions each agent can perform according to its local state ($P_i:L_i\rightarrow 2^{ACT_i}\backslash \{\emptyset\}, i\in\Omega$). In general, action performance is non-deterministic.
\end{itemize}

\begin{definition}[Interpreted system]
\label{def:IS}
	Let $\Phi$ be a set of atomic propositions and $\Omega=\{e,1,\dots,n\}$ be a set of agents. An \emph{interpreted system} $I$ is a tuple:
\begin{equation*}
	I=\langle(L_i)_{i\in\Omega}, (P_i)_{i\in\Omega}, (ACT_i)_{i\in\Omega}, S_0, \tau, \gamma\rangle
\end{equation*}
where (1) $L_i$ is the set of local states of agent $i$, and the set of global states is defined as $S=L_e\times L_1\times\dots\times L_n$ (2) $ACT_i$ is the set of actions that agent $i$ can perform, and $ACT=ACT_e\times ACT_1\times\dots\times ACT_n$ is defined as the set of joint actions (3) $S_0\subseteq S$ is the set of initial states (4) $\gamma:S\times\Phi\rightarrow \{\top, \bot\}$ is called the \emph{interpretation function} (5) $P_i:L_i\rightarrow 2^{ACT_i}\backslash \{\emptyset\}$ is the protocol for agent $i$ (6) $\tau:ACT\times S\rightarrow S$ is called the \emph{partial transition function} with the property that if $\tau(\alpha, s)$ is defined, then for all $i\in\Omega:~ACT_i(\alpha)\in P_i(l_i(s))$. We also write $s_1\xrightarrow{\alpha}s_2 $ if $\tau(\alpha, s_1) = s_2$.
\end{definition}

\label{sec:epistemicrelation}

\begin{definition}[Reachability]
	A global state $s\in S$ is \emph{reachable} in the interpreted system $I$ if there exists $s_0\in S_0$, $s_1,\dots,s_n\in S$ and $\alpha_1,\dots,\alpha_n\in ACT$ such that for all $1\leq i\leq n:~s_i=\tau(\alpha_i, s_{i-1})$ and $s=s_n$. In this paper, we use $G$ to denote the set of reachable states.
\end{definition}

For an interpreted system $I$ and each agent $i$ we define an epistemic accessibility relation on the global states as follows:

\begin{definition}[Epistemic accessibility relation]
  Let $I$ be an interpreted system and $i$ be an agent. We define the \emph{Epistemic accessibility relation for agent $i$}, written $\sim_i$,  on the global states of $I$ by $s\sim_i s' \quad\text{ iff }\quad l_i(s)=l_i(s') \mbox{ and $s$ and $s'$ are reachable}$.

\end{definition}

\subsection{CTLK logic}

We specify our properties in CTLK \cite{Lomuscio:2006:CMC}. CTL (Computational Tree Logic) is a branching-time temporal logic which has tree-like time model structure and allows quantification over paths, and CTLK adds the epistemic modality K to the CTL. CTLK is defined as follows:

\begin{definition}
Let $\Phi$ be a set of atomic propositions and $\Omega$ be a set of agents. If $p\in\Phi$ and $i\in\Omega$, then CTLK formulae are defined by:
\begin{equation*}
	\phi ::= p~|~\neg \phi~|~\phi\vee\phi~|~K_i\phi~|~EX\phi~|~EG\phi~|~E(\phi U\phi)
\end{equation*}
\end{definition}

The symbol $E$ is existential path quantifier which means ``there exists at least one path'''. Temporal connectives $X$, $G$ and $U$ mean ``neXt state'', ``all future states (Globally)'' and ``Until'''. $EX$, $EG$ and $EU$ provide the adequate set of CTLK connectives. For instance, safety properties defined by $AG(\phi)$ (all future states (Globally)) where $A$ is the universal path quantifier, can be written as $\neg E(\top U \neg\phi)$, or the equivalence for liveness properties $AF(\phi)$ (always for some future state) is $\neg EG(\neg\phi)$. Epistemic connective $K_i$ means ``agent $i$ knows that''.

\begin{example}
\label{ex:ctlk}
	Consider a conference paper review system. Assume that $\texta{a}_1$ is the author of the paper $\texta{p}_1$. Then the safety property that says if all the papers are assigned to the reviewers and $\texta{a}_2$ is the reviewer of $\texta{p}_1$, then $\texta{a}_1$ does not know the fact that $\texta{a}_2$ is the reviewer of his paper can be defined as: $AG(\texta{reviewer}(\texta{p}_1, \texta{a}_2)\rightarrow \neg K_{\texta{a}_1}\texta{reviewer}(\texta{p}_1, \texta{a}_2))$.

In an student information system, the property that states no two students can be assigned as the demonstrator of each other is specified by: $AG(\neg(\texta{demonstratorOf}(\texta{a}_2, \texta{a}_3)\wedge \texta{demonstratorOf}(\texta{a}_3, \texta{a}_2)))$.	

\end{example}

\begin{definition}[Satisfaction relation]
Let $I$ be an interpreted system, $s\in G$ where $G$ is the set of reachable states and $p\in\Phi$ where $\Phi$ is the set of atomic propositions. For any CTLK-formula $\phi$, the notation $(I,s)\models\phi$ means $\phi$ holds at state $s$ in interpreted system $I$. The relation $\models$ is defined inductively as follows:
\begin{eqnarray*}
\begin{aligned}
	&(I,s)\models p &\Leftrightarrow& \quad\gamma(s,p)=\top\\
	&(I,s)\models \neg \phi &\Leftrightarrow& \quad (I,s)\not\models \phi \\
	&(I,s)\models \phi_1\vee\phi_2 &\Leftrightarrow& \quad (I,s)\models\phi_1 \text{ or } (I,s)\models \phi_2\\
	&(I,s)\models K_i\phi &\Leftrightarrow& \quad (I,s')\models\phi \text{ for all } s'\in G \text{ such that } s\sim_i s' \\
	&(I,s)\models EX\phi &\Leftrightarrow& \quad\text{for some $s'$ such that } s\xrightarrow{\alpha} s':  (I,s')\models\phi\\
	&(I,s)\models EG\phi &\Leftrightarrow& \quad\text{there exists a path } s_1\xrightarrow{\alpha}\dots \text{ such that }s=s_0 \text{ and for all }\\
		\omit\rlap{$\qquad i\geq 0:(I,s_i)\models\phi$} \\
	&(I,s)\models E(\phi_1 U\phi_2) &\Leftrightarrow& \quad\text{there exists a path } s_1\xrightarrow{\alpha}\dots  \text{ such that }s=s_1, \text{ there is }\\
		 \omit\rlap{$\qquad$ some $i\geq 1$ such that $(I,s_i)\models\phi_2$ and for all $j<i$ we have $(I,s_j)\models\phi_1$}
\end{aligned}
\end{eqnarray*}

We use the notation $I\models\phi$ if for all $s_0\in S_0:~(I,s_0)\models\phi$.
\end{definition}


\section{Policy syntax}
\label{sec:syntax}

Multi-agent access control systems grant or deny user access to the resources and services depending on the access rights defined in the policy. Access to the resources is divided into \emph{write access}, which when granted, allows updating some system variables (in the context of this work, Boolean variables) and \emph{read access}, that returns the value of some variables when granted. In this section, we present a simple policy syntax to define actions, permissions and evolutions. In the following section, we give semantics of the policy language by constructing an interpreted system from it.

\paragraph{Technical preliminaries}
Let $V$ be a finite set of variables and $Pred$ a finite set of predicates. The notation $\vec{v}$ is used to specify a sequence of distinct variables. An \emph{atomic formula} or simply an \emph{atom} is a predicate that is applied to a sequence of variables with the appropriate length. An access control policy is a finite set of rules defined as follows: 
\begin{eqnarray*}
	\begin{aligned}
		&L::= \top~|~\bot~|~w(\vec{v})~|~L \vee L~|~L \wedge L~|~ L\rightarrow L~|~\neg L~|~\forall v\;[L]~|~\exists v\;[L] \\
		&W::= +w(\vec{v})~|~-w(\vec{v})~|~\forall v.\;W\\
		&W_s::= W~|~W_s, W\\
		&A_R::= \texta{id}(\vec{v}):\{W_s\}\leftarrow L \qquad\text{Action rule}\\
		&R_R ::= \texta{id}(\vec{v}):w(\vec{u})\leftarrow L \qquad\text{Read permission rule}
	\end{aligned}
\end{eqnarray*}
In the above, $w\in Pred$, and $w(\vec{v})$ is an atom. $L$ denotes a logical formula over atoms, which is the condition for performing an action or reading information. $\{W_s\}$ is the effect of the action that include the updates. $+w(\vec{v})$ in the effect means executing the action will set the value of $w(\vec{v})$ to \texta{true} and $-w(\vec{v})$ means setting the value to \texta{false}. In the case of $\forall v. W$ in the effect, the action updates the signed atom in $W$ for all possible values of $v$. In the case that an atom appears with different signs in multiple quantifications in the effect (for instance, $w(c,d)$ in $\forall x.+w(c,x), \forall y.-w(y,d)$), then only the sign of the last quantification is considered for the atom. \texta{id} indicates the identifier of the rule.

Let $a(\vec{v}):E\leftarrow L$ be an action rule. The \emph{free variables} of the logical formula $L$ are denoted by \textbf{fv}($L$) and are defined in the standard way. We also define \textbf{fv}$(E)=\bigcup_{e\in E}\textbf{fv}(e)$ where \textbf{fv}$(\pm w(\vec{x}))=\vec{x}$ and \textbf{fv}$(\forall x. W)$=\textbf{fv}$(W)\backslash x$. We stipulate: \textbf{fv}$(E)\cup$\textbf{fv}$(L)\subseteq \vec{v}$. If $r(\vec{v}):w(\vec{u})\leftarrow L$ is a read rule, then \textbf{fv}$(\vec{u})\cup$\textbf{fv}$(L)\subseteq \vec{v}$.

Let $\Sigma$ be a finite set objects. A \emph{ground atom} is a variable-free atom; i.e. atoms with the variables substituted with the objects in $\Sigma$. For instance, if \texta{reviewer}$\in Pred$ and \texta{Bob,Paper}$\in\Sigma$, then \texta{reviewer(Bob,Paper)} is a ground atom. In the context of this paper, we call the ground atoms as (atomic) \emph{propositions}, since they only evaluate to \texta{true} and \texta{false}. 

An \emph{action} $\alpha:\varepsilon\leftarrow\ell$ contains an identifier $\alpha$ together with the \emph{evolution rule} $\varepsilon\leftarrow\ell$, which is constructed by instantiating all the arguments in an action rule $a(\vec{v}):E\leftarrow L$ with the objects in $\Sigma$. We refer to the whole action by its identifier $\alpha$. 

In an asynchronous multi-agent system, it is crucial to know the agent that performs an action. As the convention and for the rest of this article, we consider the first argument of the action to be the agent performing that action. Therefore, in the action \texta{assignReviewer(Alice,Bob,Paper)}, \texta{Alice} is the one that assigns \texta{Bob} as the reviewer of \texta{Paper}. If $\alpha$ is an action, then \textbf{Ag}$(\alpha)$ denotes the agent that performs $\alpha$. 

A \emph{read permission} $\rho:p\leftarrow\ell$ is constructed by substituting the arguments in read permission rule $r(\vec{v}):w(\vec{u})\leftarrow L$ with the objects in $\Sigma$. $\rho$ is the identifier, $p$ is the proposition and $\ell$ is the condition for reading $p$. As for the actions, we assume the first argument in $\rho$ to be the agent that reads the proposition $p$, which is denoted by \textbf{Ag}$(\rho)$.

\begin{definition}[Policy]
\label{def:Policy}
 An \emph{access control policy} is a finite set of actions and read permissions derived by instantiating a set of rules with a finite set of objects.
\end{definition}

\section{Building an interpreted system from a policy}
\label{sec:buildis}

In access control systems, we deal with read and write access procedures. Write procedures, which update a set of variables, are contained in interpreted systems as actions. In interpreted systems, a principal knows a fact if it is included in his local state or he can deduce it by applying logical reasoning. In access control systems and in addition to the local information, agents may obtain permission to directly access some resources in the system. This permission may be granted by the system or other agents (delegation of authority). For instance, in a web application users always have access to their own profile, but they cannot access other users' profile unless the permission is granted by the owners. When a read permission to a resource is granted, the resource will become a part of agent's local state. When the permission is denied, it will be removed from agent's directly accessible information. This behaviour is similar to a system which uses dynamically changing local states to model permissions.

Interpreted systems formally contain local states which cannot change during execution of the system. In order to model temporary read permissions, we need to introduce some locally accessible information, which simulates the temporary read access. In this section, we explain how to introduce temporary read permissions when modelling access control systems. Moreover, we model access control systems in asynchronous manner using interpreted systems framework. An interpreted system is \emph{asynchronous} if all joint actions contain at most one non-$\Lambda$ agent action where $\Lambda$ denotes no-operation.

Given a policy, we build an access control system based on interpreted systems framework by considering the requirements above. Incorporating temporary read permissions requires introducing some information into the local states. We say the proposition $p$ is local to the agent $i$ if its value only depends on the local state of $i$. In the other words, for all $s,s'\in S$ where $s\sim_i s'$ we have $\gamma(s,p)=\gamma(s',p)$.

\begin{definition}[Local interpretation]
Let $L_i$ be the set of local states of agent $i$ in interpreted system $I$ and $\Phi_i$ be the set of local propositions. We define the \emph{local interpretation} for agent $i$  as a function $\gamma_i:L_i\times\Phi_i\rightarrow\{\top,\bot\}$ such that $\gamma_i(l,p)=\gamma(s,p)$ where $l_i(s)=l$ for some global state $s$. We require the set of local propositions to be pairwise disjoint.
\end{definition}

The following lemma provides the theoretical background of modelling knowledge by readability in an interpreted system.

\begin{lemma}
\label{lemma:local-copy}
Let $I$ be an interpreted system, $G$ the set of reachable states, $i$ an agent, $\Phi$ the set of propositions and $p\in\Phi$. Suppose that $p',p''\in\Phi_i$. If for all $s\in G$:
\begin{equation}
\label{eq:lem1}
	\text{if } \gamma_i(l_i(s),p'')=\top \text{ then } (I,s)\models p \;\Leftrightarrow \;\gamma_i(l_i(s),p')=\top
\end{equation}
Then we have:
\begin{equation*}
\gamma_i(l_i(s),p'')=\top \quad\Rightarrow\quad (I,s)\models K_i p \vee K_i\neg p
\end{equation*}
\end{lemma}

\begin{proof}
We first prove that
\begin{equation}
\label{eq:lempr1}
\gamma_i(l_i(s),p'')=\top \text{ and } (I,s)\models p\quad\Rightarrow\quad (I,s)\models K_i p
\end{equation}
Let us assume that $\gamma_i(l_i(s),p'')=\top$ and $(I,s)\models p$. By (\ref{eq:lem1}) we have $\gamma_i(l_i(s),p')=\top$. Consider any state $s_1\in G$ such that $s_1 \sim_i s$. By the definition of $\sim_i$, we have $l_i(s_1)=l_i(s)$. Therefore, $\gamma_i(l_i(s_1),p')=\top$ and $\gamma_i(l_i(s_1),p'')=\top$ which implies $(I,s_1)\models p$. Hence, by the definition of $K_i$ we are able to conclude that $(I,s)\models K_i p$. The proof for the second case:
\begin{equation}
\label{eq:lempr2}
\gamma_i(l_i(s),p'')=\top \text{ and } (I,s)\models \neg p\Rightarrow (I,s)\models K_i \neg p
\end{equation}
is similar to the first proof. Therefore, by (\ref{eq:lempr1}) and (\ref{eq:lempr2}) we have $\gamma_i(l_i(s),p'')=\top \Rightarrow (I,s)\models K_i p \vee K_i\neg p$.
\end{proof}

To model knowledge by readability, we incorporate all the atomic propositions that appear in the policy into the environment. We call those propositions \emph{policy propositions}. Now for each policy proposition $p$ and for each agent, we introduce two local atomic propositions: $p_{\read}$ ($p''$ in Lemma \ref{lemma:local-copy}) as the read permission of proposition $p$, and $p_{\loc}$ ($p'$ in Lemma \ref{lemma:local-copy}) as the local copy of $p$. We modify the transition function in order to satisfy the following property: for all reachable states, if $p_{\read}$ is \texta{true} (agent has read access to $p$) in a state, then $p_{\loc}$ is assigned the same value as $p$. This property guarantees agent's knowledge of proposition $p$ whenever his access to $p$ is granted.

\begin{algorithm*}
\floatname{algorithm}{Procedure}
\caption{Incorporating read permissions into evolution rules}\label{incK}
\begin{algorithmic}[1]
\Function{incKnowledge}{$\mathcal{A}_{\mathcal{C}},\mathcal{R}_{\mathcal{C}},\Phi_{\mathcal{C}}, \Sigma_{Ag}$}
\State $\triangleright$ \textbf{Input}: $\mathcal{A}_{\mathcal{C}}$ is the set of actions, $\mathcal{R}_{\mathcal{C}}$ is the set of read permissions, $\Phi_{\mathcal{C}}$ the set of policy propositions and $\Sigma_{Ag}$ the set of agents
\State $\triangleright$ \textbf{Output}: returns the updated set of actions and the set of local propositions
\State $\mathcal{A}^u_{\mathcal{C}}:=\mathcal{A}_{\mathcal{C}}$
\ForAll {$i \in \Sigma_{Ag}$}
	\State $\Phi_i:=\emptyset$ 
	\ForAll {$p \in \Phi_{\mathcal{C}}$}
		\State \textbf{determine} $r:p\leftarrow \ell_r\in \mathcal{R}_{\mathcal{C}}$ \textbf{where} \textbf{Ag}$(r)=i$
		\State $\Phi_i :=\Phi_i\cup \{p_{\loc},p_{\read}\}$
		\State $\hat{\mathcal{A}}^u_{\mathcal{C}}:=\emptyset$
		\ForAll {$\alpha:\varepsilon\leftarrow\ell \in \mathcal{A}^u_{\mathcal{C}}$} \label{line:forloop}
			\If {$+p\in\varepsilon$}
				\State \textbf{construct }$\alpha_1:\varepsilon\cup\{+p_{\loc},+p_{\read}\}\leftarrow$
				\State $\qquad\ell\wedge (\ell_r[\top/v~|+v\in\varepsilon][\bot/v'~|-v'\in\varepsilon])$ \textbf{where} $\textbf{Ag}(\alpha_1)=\textbf{Ag}(\alpha)$
				\State \textbf{construct }$\alpha_2:\varepsilon\cup\{-p_{\read}\}\leftarrow$
				\State$\qquad\ell\wedge \neg(\ell_r[\top/v~|+v\in\varepsilon][\bot/v'~|-v'\in\varepsilon])$ \textbf{where} $\textbf{Ag}(\alpha_2)=\textbf{Ag}(\alpha)$
				\State $\hat{\mathcal{A}}^u_{\mathcal{C}}:=\hat{\mathcal{A}}^u_{\mathcal{C}}\cup\{\alpha_1,\alpha_2\}$
			\ElsIf {$-p\in\varepsilon$}
				\State \textbf{construct }$\alpha_1:\varepsilon\cup\{-p_{\loc},+p_{\read}\}\leftarrow$
				\State$\qquad\ell\wedge (\ell_r[\top/v~|+v\in\varepsilon][\bot/v'~|-v'\in\varepsilon])$ \textbf{where} $\textbf{Ag}(\alpha_1)=\textbf{Ag}(\alpha)$
				\State \textbf{construct }$\alpha_2:\varepsilon\cup\{-p_{\read}\}\leftarrow$
				\State$\qquad\ell\wedge \neg(\ell_r[\top/v~|+v\in\varepsilon][\bot/v'~|-v'\in\varepsilon])$ \textbf{where} $\textbf{Ag}(\alpha_2)=\textbf{Ag}(\alpha)$
				\State $\hat{\mathcal{A}}^u_{\mathcal{C}}:=\hat{\mathcal{A}}^u_{\mathcal{C}}\cup\{\alpha_1,\alpha_2\}$
			\Else
				\If {for all $q\in\textbf{fv}(\ell_r)$ : $+q\not\in\varepsilon$ and $-q\not\in\varepsilon$} \label{line:simp}
					\State $\hat{\mathcal{A}}^u_{\mathcal{C}}:=\hat{\mathcal{A}}^u_{\mathcal{C}}\cup\{\alpha\}$
				\Else
					\State \textbf{construct }$\alpha_1:\varepsilon\cup\{+p_{\loc},+p_{\read}\}\leftarrow \ell\wedge$
					\State $\quad(\ell_r[\top/v~|+v\in\varepsilon][\bot/v'~|-v'\in\varepsilon])\wedge p\textbf{ where }\textbf{Ag}(\alpha_1)=\textbf{Ag}(\alpha)$
					\State \textbf{construct }$\alpha_2:\varepsilon\cup\{-p_{\loc},+p_{\read}\}\leftarrow \ell\wedge$
					\State $\quad(\ell_r[\top/v~|+v\in\varepsilon][\bot/v'~|-v'\in\varepsilon])\wedge \neg p\textbf{ where } \textbf{Ag}(\alpha_2)=\textbf{Ag}(\alpha)$
					\State \textbf{construct }$\alpha_3:\varepsilon\cup\{-p_{\read}\}\leftarrow \ell\wedge$
					\State $\quad\neg(\ell_r[\top/v~|+v\in\varepsilon][\bot/v'~|-v'\in\varepsilon])\textbf{ where }\textbf{Ag}(\alpha_3)=\textbf{Ag}(\alpha)$
					\State $\hat{\mathcal{A}}^u_{\mathcal{C}}:=\hat{\mathcal{A}}^u_{\mathcal{C}}\cup\{\alpha_1,\alpha_2,\alpha_3\}$
				\EndIf
			\EndIf
		\EndFor
		\State $\mathcal{A}^u_{\mathcal{C}}:=\hat{\mathcal{A}}^u_{\mathcal{C}}$
	\EndFor
\EndFor
	\State \textbf{return }$\{\Phi_i~|~i\in\Sigma_{Ag}\}$, $\mathcal{A}^u_{\mathcal{C}}$
\EndFunction
\end{algorithmic}
\label{algo:incK}
\end{algorithm*}

\paragraph{Building the interpreted system}
Given a policy $\mathcal{C}$ with $\Sigma_{Ag}$ as the set of agents, we build up an interpreted system that models the access control system in the following way:

Let $\Phi_{\mathcal{C}}$ be the set of propositions that appear in $\mathcal{C}$ (policy propositions), and $\mathcal{A}_{\mathcal{C}}$ and $\mathcal{R}_{\mathcal{C}}$ the set of actions and read permissions in $\mathcal{C}$ respectively. For an interpreted system that corresponds to the policy $\mathcal{C}$, the knowledge gained by reading system information need to be incorporated into the local states of the agents.

Procedure \ref{incK} adopts Lemma \ref{lemma:local-copy} which describes a method to model temporary read permissions. The function \RefTirName{incKnowledge} in procedure \ref{algo:incK} accepts $\mathcal{A}_{\mathcal{C}}$, $\mathcal{R}_{\mathcal{C}}$, $\Phi_{\mathcal{C}}$ and $\Sigma_{Ag}$ as the input. For each agent $i$ in $\Sigma_{Ag}$, Procedure \ref{algo:incK} generates a set of local propositions $\Phi_i$. The local state of agent $i$ consists of all valuations of $\Phi_i$. For each proposition $p\in \Phi_{\mathcal{C}}$, the set $\Phi_i$ contains two propositions $p_{\loc},p_{\read}$ where $p_{\loc}$ is the copy of $p$ and gets updated whenever $p_{\read}$ as the access permission for $p$ is \texta{true} (refer to Lemma \ref{lemma:local-copy} for the details). The procedure modifies the actions and corresponding evolutions in $\mathcal{A}_{\mathcal{C}}$ into the set $\mathcal{A}^u_{\mathcal{C}}$ in order to update the propositions in $\Phi_i$ in the appropriate way. For each action and for each agent, if $p$ appears in the effect (if-conditions in lines 12 and 18), then the action will replace with two freshly created actions: one sets $p_{\read}$ to \texta{true} and $p_{\loc}$ to the same value as $p$ if the read permission of $p$ evaluates to \texta{true} in the next state (lines 13 and 19). Otherwise (read permission of $p$ evaluates to \texta{false} in the next state), $p_{\read}$ will set to \texta{false} and $p_{\loc}$ remains unchanged (lines 15 and 21). If $p$ does not appear in the effect (line 24), $p_{\loc}$ and $p_{\read}$ will only get updated whenever the read permission of $p$ is affected by the action.

\paragraph{Calculating the symbolic transition function:}
We provide the details for calculating the symbolic transition function we use for traversing over a path in our system. The symbolic transition function accepts a set of states as input and returns the result of performing an action over the states of that set.

As a convention, we use $s[p\mapsto m]$ where $s\in S$ to denote the state that is like $s$ except that it maps the proposition $p$ to the value $m$. Let $st\subseteq S$ be a set of states. When performing the action $\alpha:\varepsilon\leftarrow\ell$ in the states of $st$, the transition is only performed in the states that satisfy the permission $\ell$. In the resulting states, the propositions that do not appear in $\varepsilon$ remain the same as in the states that the transition begins. Therefore, we define:

\begin{eqnarray*}
	\Theta_\alpha(st)=\Big\{s[p\mapsto\top\mid +p\in \varepsilon][p\mapsto\bot\mid -p\in \varepsilon] \bigm\vert s\in st, (I,s)\models \ell\Big\}
\end{eqnarray*}

\begin{definition}[Derived interpreted system]
\label{def:derivedis}
Let $\mathcal{C}$ be a policy with $\Sigma_{Ag}$ as the set of agents, $\Phi_{\mathcal{C}}$ the set of policy propositions, and $\mathcal{A}^u_{\mathcal{C}}$ and $\Phi_i$, $i\in\Sigma_{Ag}$ derived from procedure \ref{incK}. Let $\Omega=\{e\}\cup\Sigma_{Ag}$ and $\Phi=\bigcup_{i\in\Omega}\Phi_i$ where $\Phi_e=\Phi_{\mathcal{C}}$. Then the interpreted system derived from policy $\mathcal{C}$ is:

\begin{equation*}
	I_\mathcal{C}=\langle(L_i)_{i\in\Omega}, (P_i)_{i\in\Omega}, (ACT_i)_{i\in\Omega}, S_0, \tau, \gamma\rangle
\end{equation*}

where 
\begin{enumerate}
	\item $L_i$ is the set of local states of agent $i$, where each local state is a valuation of the propositions in $\Phi_i$. The set of global states is defined as $S=L_e\times L_1\times\dots\times L_n$
	\item $ACT_i=\{\alpha\in\mathcal{A}^u_{\mathcal{C}}~|~\textbf{Ag}(\alpha)=i\}\cup\{\Lambda\}$ where $\Lambda$ denotes \emph{no operation}, and a joint action is a $|\Omega|$-tuple such that at most one of the elements is non-$\Lambda$ (asynchronous interpreted system). For simplicity, \emph{we denote a joint action with its non-$\Lambda$ element} 
	\item $S_0\subseteq S$ is the set of initial states 
	\item $\gamma$ is the interpretation function over $S$ and $\Phi$. If  $p\in\Phi_i$ then we have $\gamma(s,p)=\gamma_i(l_i(s),p)$ 
	\item $P_i$  is the protocol for agent $i$ where for all $l\in L_i$: $P_i(l)=ACT_i$ 
	\item $\tau$ is the transition function that is defined as follows: if $\alpha$ is a joint action (or simply, an action) and $s\in S$, then $\tau(\alpha,s) =s'\text{ if }\Theta_{\alpha}(\{s\})=\{s'\}$.
\end{enumerate}
\end{definition}

The system that we derive from policy $\mathcal{C}$ is a special case of interpreted systems where the local states are the valuation of local propositions that are generated by the procedure \RefTirName{incKnowledge}.

\section{Abstraction technique}
\label{sec:absref}

In an interpreted system, the state space exponentially increases when extra propositions are added into the system. Considering a fragment of CTLK properties known as ACTLK as the specification language, we are able to verify the properties over an over-approximated abstract model instead of the concrete one. ACTLK is defined as follows: 

\begin{definition}
Let $\Phi$ be the set of atomic propositions and $\Omega$ set of agents. If $p\in\Phi$ and $i\in\Omega$, then ACTLK formulae are defined by:
\begin{equation*}
	\phi ::= p~|~\neg p~|~\phi\wedge\phi~|~\phi\vee\phi~|~K_i\phi~|~AX\phi~|~A(\phi U\phi)~|~A(\phi R\phi)
\end{equation*}
\end{definition}

where the symbol $A$ is universal path quantifier which means ``for all the paths''. 

To provide a relation between the concrete model and the abstract one, we extend the \emph{simulation relation} introduced in \cite{ClarkeMC} to cover the epistemic relation between states. Using the abstraction technique that preserves simulation relation between the concrete model and the abstract one, we are able to verify ACTLK specification formulas over the model. In this paper and for abstraction and refinement, we focus on safety properties expressed in ACLK. The advantages of safety properties are first, they are capable of expressing \emph{policy invariants}, and second, the generated counterexample contains finite sequence of actions (or transitions). We can extend the abstraction refinement method to the full ACTLK by unfolding the loops in the counterexamples into finite transitions as described in \cite{CEGAR}, which is outside the scope of this paper.

\subsection{Existential abstraction}

The general framework of existential abstraction was first introduced by Clark et. al in \cite{ClarkeMC}. Existential abstraction partitions the states of a model into clusters, or equivalence classes. The clusters form the states of the abstract model. The transitions between the clusters in the abstract model give rise to an over-approximation of the original (or concrete) model that \emph{simulates} the original one. So, when a specification in ACTL (or in the context of this paper, ACTLK) logic is true in the over-approximated model, it will be true in the concrete one. Otherwise, a counterexample will be generated which needs to be verified over the concrete model.

\begin{notation}
For simplicity, we use the same notation ($\sim_i$) for the epistemic accessibility relation in both the concrete and abstract interpreted systems.
\end{notation}

\begin{definition}[Simulation]
\label{def:simulation}
Let $I$ and $\tilde{I}$ be two interpreted systems, $\Omega$ be the set of agents in both systems, and $\Phi$ and $\tilde{\Phi}$ the corresponding set of propositions where $\tilde{\Phi}\subseteq\Phi$. The relation $H\subseteq S\times \tilde{S}$  is \emph{simulation relation} between $I$ and $\tilde{I}$ if and only if:

\begin{enumerate}
	\item \label{item:sim1}For all $s_0\in S_0$, there exists $\tilde{s}_0\in\tilde{S_0}$ st. $(s_0,\tilde{s}_0)\in H$.
\end{enumerate}

and for all $(s,\tilde{s})\in H$:

\begin{enumerate}
\setcounter{enumi}{1}
	\item \label{item:sim2}For all $p\in\tilde{\Phi}:~\gamma(s,p)=\tilde{\gamma}(\tilde{s},p)$
	\item \label{item:sim3}For each state $s'\in S$ such that $\tau(s,\alpha)=s'$ for some $\alpha\in ACT$, there exists $\tilde{s}'\in\tilde{S}$ and $\tilde{\alpha}\in \tilde{ACT}$ such that $\tilde{\tau}(\tilde{s},\tilde{\alpha})=\tilde{s}'$ and $(s',\tilde{s}')\in H$.
	\item \label{item:sim4}For each state $s'\in S$ such that $s\sim_i s'$, there exists $\tilde{s}'\in\tilde{S}$ such that $\tilde{s}\sim_i \tilde{s}'$ and $(s',\tilde{s}')\in H$.
\end{enumerate}
\end{definition}

The above definition for simulation relation over the interpreted systems is similar to the one for Kripke model \cite{CEGAR}, except that the relation for the epistemic relation is introduced. If such simulation relation exists, we say that $\tilde{I}$ \emph{simulates} $I$ (denoted by $I \preceq \tilde{I}$).

If $H$ is a function, that is, for each $s\in S$ there is a unique $\tilde{s}\in\tilde{S}$ such that $(s,\tilde{s})\in H$, we write $h(s)=\tilde{s}$ instead of $(s,\tilde{s})\in H$.

\begin{lemma}
\label{lem:equivpath}
Let $I \preceq \tilde{I}$, $s_1\in S$, $\tilde{s}_1\in\tilde{S}$ and $(s_1,\tilde{s}_1)\in H$ where $H$ is the simulation relation between $I$ and $\tilde{I}$. Then for each path $s_1\xrightarrow{\alpha_2}\dots$ in $I$, there exists a path $\tilde{s}_1\xrightarrow{\tilde{\alpha}_2}\dots$ in $\tilde{I}$ such that for all $i\geq 1$, $(s_i,\tilde{s}_i)\in H$ holds. 
\end{lemma}

\begin{proof}
The proof is trivial by item \ref{item:sim3} in definition \ref{def:simulation} and induction over the state transitions.
\end{proof}

\begin{proposition}
\label{prop:verification}
For every ACTLK formula $\varphi$ over propositions $\tilde{\Phi}$, if $I \preceq \tilde{I}$ and $\tilde{I}\models\varphi$, then $I\models\varphi$.
\end{proposition}
\begin{proof}
	To prove the proposition, we first prove if $I \preceq \tilde{I}$ and $H$ is the simulation relation, then for all $\tilde{s}\in \tilde{S}$ and $s\in S$ where $(s,\tilde{s})\in H$, $(\tilde{I},\tilde{s})\models\varphi$ implies  $(I,s)\models\varphi$. We assume $\varphi$ is in NNF. The proof proceeds by induction over the structure of $\varphi$. Let $s\in S$, $\tilde{s}\in\tilde{S}$ and $(s,\tilde{s})\in H$.

\begin{itemize}
	\item If $(\tilde{I},\tilde{s})\models p$ where $p$ an atomic formula, then $\gamma(\tilde{s},p)=\top$. By item \ref{item:sim2} in definition \ref{def:simulation} we have $\gamma(s,p)=\top$ which implies $(I,s)\models p$. The case is similar for $\varphi=\neg p$.

	\item If $(\tilde{I},\tilde{s})\models \varphi_1\wedge\varphi_2$, then $(\tilde{I},\tilde{s})\models \varphi_1$ and $(\tilde{I},\tilde{s})\models \varphi_2$. By induction hypothesis we have $(I,s)\models \varphi_1$ and $(I,s)\models \varphi_2$. Therefore, $(I,s)\models \varphi_1\wedge\varphi_2$. The case is similar for $\varphi=\varphi_1\vee\varphi_2$.

	\item Assume $(\tilde{I},\tilde{s})\models AX\varphi_1$. If $s\xrightarrow{\alpha} s'$ is a path in $I$, then by Lemma \ref{lem:equivpath} there exists a path $\tilde{s}\xrightarrow{\tilde{\alpha}} \tilde{s}'$ in $\tilde{I}$ where $(s',\tilde{s}')\in H$. By the assumption we have $(\tilde{I},\tilde{s}')\models \varphi_1$. Then the induction hypothesis implies $(I,s')\models \varphi_1$. Thus we can conclude that $(I,s)\models AX\varphi_1$.

	\item Assume $(\tilde{I},\tilde{s})\models A(\varphi_1 U\varphi_2)$. Let $s_1\xrightarrow{\alpha_2}\dots$ be a path in $I$ where $s_1=s$ and $\tilde{s}_1\xrightarrow{\tilde{\alpha}_2}\dots$ the corresponding path in $\tilde{I}$ where $\tilde{s}_1=\tilde{s}$. By the assumption, there exists some $i\geq 1$ where $(\tilde{I},\tilde{s}_i)\models\varphi_2$ and $(\tilde{I},\tilde{s}_i)\models\varphi_1$ for all $j<i$. By induction hypothesis and Lemma \ref{lem:equivpath}, $(I,s)\models \varphi_1 U\varphi_2$. As this property holds for all the path starting at $s$, we can conclude $(I,s)\models A(\varphi_1 U\varphi_2)$.

	\item Assume $(\tilde{I},\tilde{s})\models A(\varphi_1 R\varphi_2)$. The proof is similar to the case for $(\tilde{I},\tilde{s})\models A(\varphi_1 U\varphi_2)$.

	\item Assume $(\tilde{I},\tilde{s})\models K_i\varphi$. We pick a state $s'\in S$ where $s'\sim_i s$. By item \ref{item:sim4} in definition \ref{def:simulation}, there exists $\tilde{s}'\in \tilde{S}$ where $\tilde{s}'\sim_i \tilde{s}$ and $(s',\tilde{s}')\in H$. By the assumption, $(\tilde{I},\tilde{s}')\models \varphi$. Induction hypothesis implies that $(I,s')\models \varphi$. As this property holds for all the states with accessibility relation $\sim_i$ to $s$, we have  $(I,s')\models K_i\varphi$.
\end{itemize}

Now, if $\tilde{I}\models\varphi$ or in the other words, for all $\tilde{s}_0\in\tilde{S}_0$: $(\tilde{I},\tilde{s})\models\varphi$, then by item \ref{item:sim1} in definition \ref{def:simulation} and the above proof we have for all $s_0\in S_0$: $(I,s)\models\varphi$ or equivalently $I\models\varphi$.

\end{proof}

\subsection{Variable hiding abstraction}
\label{sec:varhide}

Variable hiding is a popular technique in the category of existential abstraction. In our methodology, we consider factorizing the concrete state space into equivalence classes that act as abstract states by abstracting away a set of system propositions. In our approach, the states in each equivalence class are only different in the valuation of the hidden propositions. Also the transitions between the states of the abstract model are defined in such a way that the abstract model simulates the concrete one. Our refinement procedure will be splitting the abstract states by putting back some of the atomic proportions that were hidden in the abstract model. We refine the model by analysing the counterexample generated when verifying safety properties described in ACTLK logic. The model checker will output a counterexample if the property does not hold.

\begin{definition}(Local state relation)
\label{def:locrelation}
Let $I_\mathcal{C}$ be an interpreted system derived from policy $\mathcal{C}$, $L_i$ and $\Phi_i$ be the set of local states and local propositions for the agent $i$, and $\tilde{\Phi}_i\subseteq\Phi_i$. The local relation $\Re_i$ is defined as:
\begin{equation*}
		\text{for all } l_1,l_2\in L_i: \qquad l_1\Re_i l_2\quad\text{ iff }\quad\text{for all }p\in \tilde{\Phi}_i: \gamma_i(l_1,p)=\gamma_i(l_2,p)
\end{equation*}
where $\gamma_i$ is the local interpretation for the agent $i$. The function $h_i:L_i\rightarrow L_i/\Re_i$  is the surjection which maps elements of $L_i$ into equivalence classes of $\Re_i$.
\end{definition}

\begin{definition}[Action classification]
\label{def:actionabs}
Let $\alpha:\varepsilon\leftarrow\ell \in ACT$ and $\tilde{\Phi}\subseteq\Phi$. We define $\alpha':\varepsilon'\leftarrow\ell' \in[\alpha]$ iff $\{\pm p\in\varepsilon'~|~p\in \tilde{\Phi}\}=\{\pm p\in\varepsilon~|~p\in \tilde{\Phi}\}$, $\exists (\Phi\backslash\tilde{\Phi}).\ell'\equiv\exists (\Phi\backslash\tilde{\Phi}).\ell$ and $\textbf{Ag}(\alpha')=\textbf{Ag}(\alpha)$.
\end{definition}

In the above definition, the infix notation $\equiv$ denotes the semantically equivalence relation. Formally $\exists x.f$ for a Boolean function $f$ is defined as $f[0/x]\vee f[1/x]$ which means $f$ could be made to true by putting $x$ to 0 or to 1. If $X=\{x_1,\dots,x_n\}$, then $\exists X.f=\exists x_1\dots\exists x_n.f$.

\begin{definition}[Abstract interpreted system]
\label{def:abspolicy}
Given a policy $\mathcal{C}$, let $\Omega,\Phi$ and $\mathcal{A}^u_{\mathcal{C}}$ be deduced as described in section \ref{sec:buildis} and $I_{\mathcal{C}}$ be the derived interpreted system.
Let $\tilde{\Phi}\subseteq\Phi$ and $\tilde{\Omega}=\Omega$. We define Interpreted system $\tilde{I}_{\mathcal{C}}$ as:
\begin{equation*}
	\tilde{I}_\mathcal{C}=\langle(\tilde{L}_i)_{i\in\tilde{\Omega}}, (\tilde{P}_i)_{i\in\tilde{\Omega}}, (\tilde{ACT}_i)_{i\in\tilde{\Omega}}, \tilde{S}_0, \tilde{\tau}, \tilde{\gamma}\rangle
\end{equation*}

\noindent where 
\begin{enumerate}
	\item $\tilde{L}_i=L_i/\Re_i$ where $\Re_i$ is defined in definition \ref{def:locrelation} over $L_i$, and $\tilde{S}=\tilde{L}_e\times \tilde{L}_1\times\dots\times \tilde{L}_n$ 
	
	\item $\tilde{ACT}_i=\{[\alpha]\mid \alpha\in \mathcal{A}^u_{\mathcal{C}}\text{ and }\textbf{Ag}(\alpha)=i\}$ and a joint action is a $|\tilde{\Omega}|$-tuple such that at most one of the elements is non-$\Lambda$ - i.e. the system is asynchronous. As before, each joint action is shown by its non-$\Lambda$ element. If $\tilde{\alpha}=[\alpha]$, then the evolution rule for $\tilde{\alpha}$ is $\tilde{\varepsilon}\leftarrow\tilde{\ell}$ where $\tilde{\varepsilon}=\{\pm p\in\varepsilon~|~p\in \tilde{\Phi}\}$ and $\tilde{\ell}=\exists (\Phi\backslash\tilde{\Phi}).\ell$ 
	
	\item $\tilde{S}_0=\{(h_i(l_i(s)))_{i\in\tilde{\Omega}}\mid s\in S_0\}$ where $h_i$ as in definition \ref{def:locrelation} maps the elements of $L_i$ to $\tilde{L}_i$

	\item For all $\tilde{l}\in\tilde{L}_i$ and for all $p\in\tilde{\Phi}_i$ we have $\tilde{\gamma}_i(\tilde{l},p)=\gamma_i(l,p)$ where $\tilde{l}=h_i(l)$ 
	
	\item $\tilde{P}_i$  is the protocol for agent $i$ where for all $\tilde{l}\in \tilde{L}_i$: $\tilde{P}_i(\tilde{l})=\tilde{ACT}_i$ 
	
	\item $\tilde{\tau}$ is the transition function defined as follows: If $\tilde{\alpha}$ is a joint action, $\tilde{s}\in \tilde{S}$ and $\tilde{\Theta}_{\tilde{\alpha}}$ is the symbolic transition function for interpreted system $\tilde{I}_\mathcal{C}$ and action $\tilde{\alpha}$, then $\tilde{\tau}(\tilde{\alpha},\tilde{s}) =\tilde{s}'\;\text{ if }\;\tilde{\Theta}_{\tilde{\alpha}}(\{\tilde{s}\})=\{\tilde{s}'\}$

\end{enumerate}
\end{definition}

\begin{proposition}
\label{prop:policysim}
If $I_\mathcal{C}$ is the interpreted system derived from policy $\mathcal{C}$ and $\tilde{I}_\mathcal{C}$ is defined as in definition \ref{def:abspolicy}, then $I_{\mathcal{C}} \preceq \tilde{I}_{\mathcal{C}}$.
\end{proposition}

\begin{proof}
Let $h:S\rightarrow \tilde{S}$ be a function where $h(s)=(h_i(l_i(s)))_{i\in\tilde{\Omega}}$ and $h_i$ is defined as in definition \ref{def:locrelation}. We show that $\tilde{I}_{\mathcal{C}}$ simulates $I_{\mathcal{C}}$ under $h$. Item \ref{item:sim1} in definition \ref{def:simulation} trivially holds by property (3). Item \ref{item:sim2} holds by property (4) and the fact that if $p\in\tilde{\Phi}$, then there is an agent $i$ where $p\in\tilde{\Phi}_i$ and we have $\tilde{\gamma}(\tilde{s},p)=\tilde{\gamma}_i(\tilde{l}_i(\tilde{s}),p)$.

Now assume that $h(s)=\tilde{s}$ and $\tau(\alpha,s)=s'$, which is equivalent to $\Theta_{\alpha}(\{s\})=\{s'\}$. If $\alpha:\varepsilon\leftarrow\ell$ can be performed in $s$, then we have $(I,s)\models\ell$. It is trivial to show that $(I,s)\models\exists(\Phi\backslash\tilde{\Phi}). \ell$ using structural induction. Since the formula $\exists (\Phi\backslash\tilde{\Phi}). \ell$ only contains the propositions in $\tilde{\Phi}$, then by item  \ref{item:sim2} in definition \ref{def:simulation} we have $(\tilde{I},\tilde{s})\models \exists (\Phi\backslash\tilde{\Phi}). \ell$. Let $\tilde{\alpha}=[\alpha]$. By definition \ref{def:actionabs}, $\tilde{\alpha}$ can be performed in $\tilde{s}$. From $\tilde{\varepsilon}\subseteq\varepsilon$ we infer that the performance of $\tilde{\alpha}$ on $\tilde{s}$ results in a state $\tilde{s}'$ where all the propositions in $\tilde{\Phi}$ have the same value in $\tilde{s}'$ as in $s'$. Hence, $h(s')=\tilde{s}'$ as required for item \ref{item:sim3} in definition \ref{def:simulation}.

Let us assume that $h(s)=\tilde{s}$ and $s\sim_i s'$. Therefore $l_i(s)=l_i(s')$ which means that for all $p\in\Phi_i:\; \gamma(s,p)=\gamma(s',p)$. Since $\tilde{\Phi}\subseteq\Phi$, then $\tilde{\Phi}_i\subseteq\Phi_i$. By item  \ref{item:sim2} in definition \ref{def:simulation}, for all $p\in\tilde{\Phi}_i:\; \gamma(s,p)=\tilde{\gamma}(\tilde{s},p)$. Let us assume that $h(s')=\tilde{s}'$. Then for all $p\in\tilde{\Phi}_i:\; \gamma(s',p)=\tilde{\gamma}(\tilde{s}',p)$. Hence we have for all $p\in\tilde{\Phi}_i:\; \tilde{\gamma}(\tilde{s},p)=\tilde{\gamma}(\tilde{s}',p)$. Therefore $\tilde{s}\sim_i\tilde{s}'$ as required for item \ref{item:sim4}.
\end{proof}

\begin{definition}
We define $h_A:ACT\rightarrow\tilde{ACT}$ as the surjection that maps the actions in the concrete model to the actions in the abstract one.
\end{definition}

Given a policy, by using Proposition \ref{prop:policysim} we can build up an abstract access control system by hiding a set of propositions and abstracting the evolution rules. Now by proposition \ref{prop:verification}, it is possible to verify ACTLK properties over the abstract model, and refine the abstraction, if the property does not hold and the counterexample is found to be spurious.

\section{Automated refinement}
\label{sec:ref}

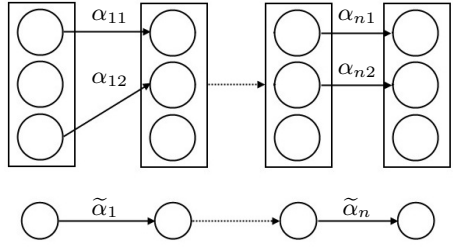
\begin{figure}[t]
	\centering
	\begin{overpic}[width=6cm]{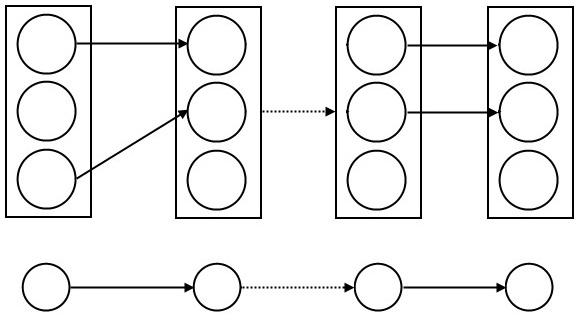}
		\put(19,50){$\alpha_{11}$}		
		\put(73,50){$\alpha_{n1}$}		
		\put(19,36){$\alpha_{12}$}		
		\put(73,38){$\alpha_{n2}$}		
		\put(19,7){$\tilde{\alpha}_{1}$}		
		\put(74,7){$\tilde{\alpha}_{n}$}		
	\end{overpic}
	\caption{The counterexample provided by the abstract model may not be valid on the concrete one. The labels represent the actions that result in the transitions.}
\label{fig:spurious}
\end{figure}

Our counterexample based abstraction refinement method consists of three steps:
\begin{itemize}
	\item \emph{Generating the initial abstraction}: It is done by examining transition blocks corresponding to the variables and constructing clusters of variables which interfere with each other via transition conditions. In our approach, we build the simplest possible initial abstract model by only retaining only the propositions appear in specification $\varphi$ that we aim to verify.
	\item \emph{Model-checking the abstract structure}: Model-checking will be performed on the abstract model for a specification $\varphi$. If the abstract model satisfies $\varphi$, then it can be concluded that the concrete model also satisfies $\varphi$. If the abstract model checking generates a counterexample, it should be checked if the counterexample is an actual counterexample for the concrete model. If it is a spurious counterexample in the concrete model as in figure \ref{fig:spurious}, the abstract system should be refined by proceeding to the next step.
	\item \emph{Refining the abstraction}: The counterexample guided framework refines the abstract model by partitioning the states in abstract model in such a way that the refined model does not admit the same counterexample. For the refinement, we turn some of the invisible variables into visible. After refinement of the abstract model, step 2 will be proceeded.
\end{itemize}

The process of abstraction and refinement will eventually terminate, as in the worst case, the refined model becomes the same as the concrete one, which is a finite state model. Therefore in the worst case, the verification will turn into the verification of the concretised model.

\subsection{Generating the initial abstraction}

For automatic abstraction refinement, we build the initial model as simple as possible. For an ACTLK formula $\varphi$, we keep all the atomic propositions that appear in $\varphi$ visible in the abstract model and hide the rest. The abstract model is built up by definition \ref{def:abspolicy}.

\subsection{Validation of counterexamples}
\label{sec:validation}

The structure of a counterexample created by the verification of an ACTLK formula is different from the counterexample generated in the absence of knowledge modality. In an ACTLK counterexample, we have epistemic relations as well as temporal ones. Analysis of such counterexamples is more complicated than the counterexamples for temporal properties.

A counterexample for a safety property in ACTLK is a loop-free tree-like graph with states as vertices, and temporal and epistemic transitions as edges. Every counterexample has an initial state as the root. A temporal transition in the graph is labelled with its corresponding action and epistemic transition is labelled with the corresponding epistemic relation. We define a \emph{temporal path} as a path that contains only temporal transitions. An \emph{epistemic path} contains at least one epistemic transition. Every state in the counterexample is \emph{reachable from an initial state} in the model, which may differ from the root. For any state $s$, we write also $s$ for the empty path which starts and finishes in $s$.

\begin{figure}
	\centering
	\begin{overpic}[width=5cm]{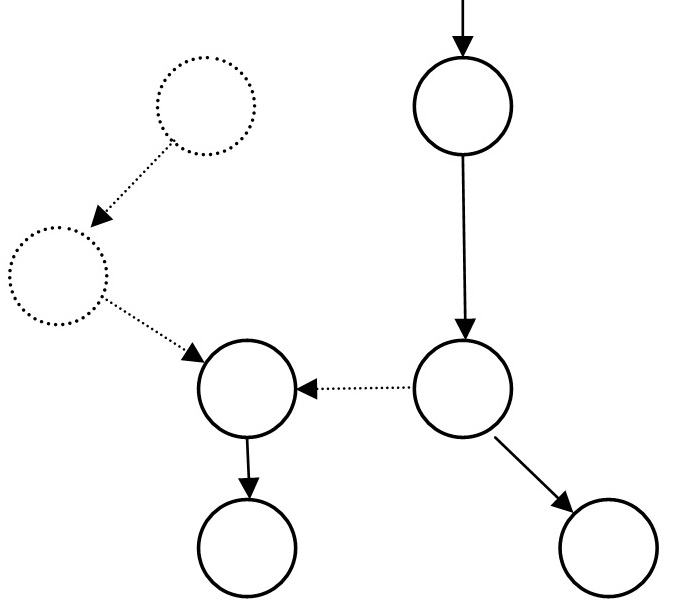}
		\put(64,71){$\tilde{s}_0$}		
		\put(70,51){$\tilde{\alpha}_1$}		
		\put(64,29){$\tilde{s}_1$}		
		\put(78,20){$\tilde{\alpha}_2$}		
		\put(84,6){$\tilde{s}_3$}		
		\put(27,71){$\tilde{s}'_0$}	
		\put(10,65){$\tilde{\alpha}'_1$}
		\put(5,46){$\tilde{s}'_1$}	
		\put(12,34){$\tilde{\alpha}'_2$}
		\put(32,30){$\tilde{s}'_2$}	
		\put(40,18){$\tilde{\alpha}'_3$}
		\put(32,6){$\tilde{s}'_3$}	
		\put(47,35){$\sim_a$}
	\end{overpic}
	\caption{A tree-like counterexample generated by the verification of an ACTLK safety property over the abstract model. In the diagram, $\tilde{s}_0, \tilde{s}'_0\in S_0$ and $\tilde{s}_1\sim_a \tilde{s}'_2$. As reachability is a requirement for $\tilde{s}_1\sim_a \tilde{s}'_2$ and $\tilde{s}_1$ is already reachable, the temporal path $\tilde{s}'_0\xrightarrow{\tilde{\alpha}'_1}\tilde{s}'_1\xrightarrow{\tilde{\alpha}'_2}\tilde{s}'_2$ provides the witness for the reachability of $\tilde{s}'_2$. Considering this witness is required in counterexample checking.}
	\label{fig:cte}
\end{figure}
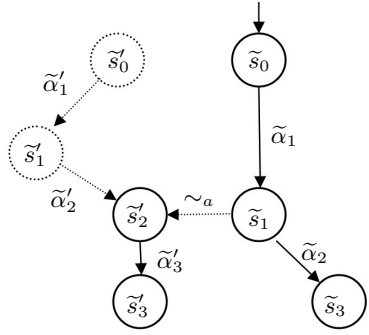

\textbf{Counterexample formalism:} A tree is a finite set of temporal and epistemic paths with an initial state as the root. Each path begins from the root and finishes at a leaf. For an epistemic transition over a path, we use the same notation as the epistemic relation while we consider the transition to be from left to the right. For instance, the tree in the figure \ref{fig:cte} is formally presented by:
\begin{equation*}
\{\tilde{s}_0\xrightarrow{\tilde{\alpha}_1}\tilde{s}_1\xrightarrow{\tilde{\alpha}_2}\tilde{s}_3,\; \tilde{s}_0\xrightarrow{\tilde{\alpha}_1}\tilde{s}_1\sim_a\tilde{s}'_2\xrightarrow{\tilde{\alpha}'_3}\tilde{s}'_3\}
\end{equation*}

To verify a tree-like counterexample, we traverse the tree in a \emph{depth-first} manner. An abstract counterexample is valid in the concrete model if a real counterexample in the concrete model corresponds to it.

We use the notation $s\rightarrow s'$ when the type of the transition from $s$ to $s'$ is not known.

\begin{definition}[Vertices, root]
Let $\tilde{ce}$ be a counterexample. Then \textbf{Vert}$(\tilde{ce})$ denotes the set of all the states that appear in $\tilde{ce}$. \textbf{Root}$(\tilde{ce})$ denotes the root of $\tilde{ce}$. For a path $\tilde{\pi}$, \textbf{Root}$(\tilde{\pi})$ denotes the state that $\tilde{\pi}$ starts with.
\end{definition}

\begin{definition}[Corresponding paths]
\label{def:corr-path}
Let $\tilde{I}$ be an abstract model of the interpreted system $I$, $h$ be the abstraction function, and $h_A$ be the function that maps the actions in $I$ to the ones in $\tilde{I}$. The concrete path $\pi=s_1\rightarrow\dots\rightarrow s_n$ in the concrete model corresponds to the path $\tilde{\pi}=\tilde{s}_1\rightarrow\dots\rightarrow\tilde{s}_n$ in the abstract model, if

	\begin{itemize}
		\item For all $1\leq i\leq n:\,\tilde{s}_i=h(s_i)$
		\item If $\tilde{s}_i \xrightarrow{\tilde{\alpha}_{i+1}}\tilde{s}_{i+1}$ is a temporal transition, we have $s_i\xrightarrow{\alpha_{i+1}}s_{i+1}$ where $h_A(\alpha_{i+1})=\tilde{\alpha}_{i+1}$.
		\item If $\tilde{s}_i\sim_a\tilde{s}_{i+1}$ is an epistemic transition, then $s_i\sim_a s_{i+1}$ and $s_{i+1}$ is reachable in the concrete model.
	\end{itemize}
\end{definition}

\begin{definition}[Concrete counterexample]
\label{def:conc-ce}
Let $\tilde{ce}$ be a tree-like counterexample in the abstract model where $\textbf{Root}(\tilde{ce})\in \tilde{S}_0$. A concrete counterexample $ce$ corresponds to $\tilde{ce}$ if $\textbf{Root}(ce)\in S_0$ and there exists a one-to-one correspondence between the states and the paths of the counterexamples $ce$ and $\tilde{ce}$ according to the definition \ref{def:corr-path}.
\end{definition}

\begin{figure*}[t]
\begin{mathpar}
	\inferrule* [left=TemporalCheck]
		{h_A^{-1}(\tilde{\alpha})=\{\alpha_1,\dots ,\alpha_n\}}
		{(\tilde{s}\xrightarrow{\tilde{\alpha}}\tilde{s}'\,||\,\pi,st) \Rightarrow_t (\pi,\bigcup_{i=1}^{n}\Theta_{\alpha_i}(st)\cap h^{-1}(\tilde{s}'))}
\end{mathpar}
\begin{mathpar}
	\inferrule* [left=EpistemicCheck]
		{\pi'=\tilde{s}'_0\xrightarrow{\tilde{\alpha}'_1}\dots\xrightarrow{\tilde{\alpha}'_m}\tilde{s}' \text{ is a temporal path to }\tilde{s}'\text{ where } \tilde{s}'_0\in \tilde{S}_0\\
		(\pi',S_0\cap h^{-1}(\tilde{s}'_0))\Rightarrow^*_t (\tilde{s}',st')\\
		\hat{st}=\{s\in st'~|~l_a(s)\in L_a(st)\}
}
		{(\tilde{s}\sim_a\tilde{s}'\,||\,\pi,st) \Rightarrow_e (\pi,\hat{st})}
\end{mathpar}
	\caption{Temporal and epistemic transition rules. In \RefTirName{EpistemicCheck} rule, $\pi'$ is the witness for the reachability of $\tilde{s}'$ in the abstract model, and $st'$ is the concrete states that are reachable through the concrete paths corresponding to $\pi'$. In the case that the model-checker returns all the abstract paths to $\tilde{s}'$, let us say $\tilde{\Pi}'$, then $st'$ will be calculated as $st'=\bigcup\{st\mid \pi'=\tilde{s}'_0\rightarrow\dots\rightarrow\tilde{s}'\in\tilde{\Pi}', \tilde{s}'_0\in\tilde{S}_0 \text{ and }(\pi',S_0\cap h^{-1}(\tilde{s}'_0))\Rightarrow^*_t (\tilde{s}',st)\}$.}

	\label{fig:forwardC}
\end{figure*}

To \emph{verify a path} in the counterexample, we define two transition rules \RefTirName{TemporalCheck} and \RefTirName{EpistemicCheck} denoted by $\Rightarrow_t$ and $\Rightarrow_e$ as in figure \ref{fig:forwardC}. For a path with the transition $\tilde{s}\xrightarrow{\tilde{\alpha}}\tilde{s}'$ as the head and for the concrete states $st$, the rule $\Rightarrow_t$ finds all the successors of the states in $st$ which reside in $h^{-1}(\tilde{s}')$. If the head of the path is the epistemic transition $\tilde{s}\sim_a \tilde{s}'$, then the rule $\Rightarrow_e$ extracts all the \emph{reachable states} in $h^{-1}(\tilde{s}')$ corresponding to $\pi'$ as the witness of reachability of $\tilde{s}'$, which has common local states with some states in $st\subseteq h^{-1}(\tilde{s})$. Both the temporal and epistemic rules are deterministic.

\begin{definition}
We write $\Rightarrow^*_t$ to denote a sequence of temporal transitions $\Rightarrow_t$. We use $\Rightarrow^*$ to denote a sequence of the transitions $\Rightarrow_t$ or $\Rightarrow_e$.
\end{definition}

\begin{proposition}[Soundness of $\Rightarrow^*_t$]
\label{prop:temp-path}
Let $\tilde{\pi}$ be a temporal path in the abstract model which starts at $\tilde{s}_1$ and ends in $\tilde{s}_n$. If $st_1\subseteq h^{-1}(\tilde{s}_1)$ and $(\tilde{\pi},st_1)\Rightarrow^*_t (\tilde{s}_n,st_n)$  for some $\emptyset\subset st_n\subseteq S$, then there exists a concrete path that starts from a state in $st_1$ and ends in a state in $st_n$.
\end{proposition}

\begin{proof}
We use induction over the length of the path. 

\textbf{Base case:} $\tilde{\pi}=\tilde{s}_1$. Then there is no transition from $(\tilde{s}_1,st_1)$ and therefore,  the concrete path is a state in $st_1$.

\textbf{Inductive case:} Assume by inductive hypothesis that for all $\tilde{\pi}=\tilde{s}_i\xrightarrow{\tilde{\alpha}_{i+1}}\dots\xrightarrow{\tilde{\alpha}_{i+k}}\tilde{s}_{i+k}$ of length $k$, if $(\tilde{\pi},st_i)\Rightarrow^*_t (\tilde{s}_{i+k},st_{i+k})$ for some $st_i,st_{i+k}\subseteq S$, then there exists a concrete path which begins at a state in $st_i$ and ends in a state in $st_{i+k}$. Consider that $\tilde{\pi}'=\tilde{s}_{i-1}\xrightarrow{\tilde{\alpha}_i}\tilde{s}_{i}\,||\,\tilde{\pi}$ is a path of the length $k+1$ where  $(\tilde{s}_{i-1}\xrightarrow{\tilde{\alpha}_i}\tilde{s}_{i}\,||\,\tilde{\pi},st_{i-1})\Rightarrow_t (\tilde{\pi},st_i) \Rightarrow^*_t (\tilde{s}_{i+k},st_{i+k})$. By induction hypothesis, there exists a concrete path that begins at some state $s_i\in st_i$ and ends in $s_{i+k}\in st_{i+k}$. By the definition of $\Rightarrow_t$, every state in $st_i$ is the successor of some states in $st_{i-1}$. Therefore, there exists $s_{i-1}\in st_{i-1}$ and $\alpha_i\in h_A^{-1}(\tilde{\alpha}_i)$ such that $\{s_i\}=\Theta_{\alpha_i}(\{s_{i-1}\})$. So we select the corresponding transition in the concrete model to be $s_{i-1} \xrightarrow{\alpha_i} s_i$ which allows $s_{i-1}$ to reach $s_{i+k}$ by the existence of a concrete path from $s_i$ to $s_{i+k}$. 
\end{proof}

By proposition \ref{prop:temp-path} and definition \ref{def:conc-ce},  if $\tilde{\pi}=\tilde{s}_0\xrightarrow{\tilde{\alpha}_1}\dots\xrightarrow{\tilde{\alpha}_n}\tilde{s}_n$ is a path in the counterexample where $(\tilde{\pi},S_0\cap h^{-1}(\tilde{s}_0))\Rightarrow^*_t (\tilde{s}_n,st_n)$, then there exists a corresponding concrete path beginning at an initial state $s_0\in S_0\cap h^{-1}(\tilde{s}_0)$ which ends at some state $s_n\in st_n$.

\begin{proposition}[Soundness of $\Rightarrow^*$]
\label{prop:soundnessoftemp}
Let $\tilde{\pi}=\tilde{s}_1\rightarrow\dots\rightarrow\tilde{s}_n$ be a path in the abstract model. If $st_1\subseteq h^{-1}(\tilde{s}_1)$ and $(\tilde{\pi},st_1)\Rightarrow^* (\tilde{s}_n,st_n)$  for some $\emptyset\subset st_n\subseteq S$, then there exists a concrete path that starts from a state in $st_1$ and ends in a state in $st_n$.
\end{proposition}

\begin{proof}
For the general form of a path that contains both temporal and epistemic transitions, we use the similar approach as in proposition \ref{prop:temp-path}. 

\textbf{Base case:} $\tilde{\pi}=\tilde{s}_1$. Then there is no transition from $(\tilde{s}_1,st_1)$ and therefore,  the concrete path is a state in $st_1$.

\textbf{Inductive case:} Assume by inductive hypothesis that for all $\tilde{\pi}=\tilde{s}_i\rightarrow\dots\rightarrow\tilde{s}_{i+k}$ of length $k$, if $(\tilde{\pi},st_i)\Rightarrow^* (\tilde{s}_{i+k},st_{i+k})$ for some $st_i,st_{i+k}\subseteq S$, then $\tilde{\pi}$ has a corresponding concrete path which begins at a state in $st_i$ and ends in a state in $st_{i+k}$.

\begin{itemize}
	\item Consider that $\tilde{\pi}'=\tilde{s}_{i-1}\xrightarrow{\tilde{\alpha}_i}\tilde{s}_{i}\,||\,\tilde{\pi}$ is a path of length $k+1$ where  $(\tilde{s}_{i-1}\xrightarrow{\tilde{\alpha}_i}\tilde{s}_{i}\,||\,\tilde{\pi},st_{i-1})\Rightarrow_t (\tilde{\pi},st_i) \Rightarrow^* (\tilde{s}_{i+k},st_{i+k})$. By induction hypothesis, there exists a concrete path that begins at some state $s_i\in st_i$ and ends in $s_{i+k}\in st_{i+k}$. By the same analysis as in the proof of proposition \ref{prop:temp-path}, there exists $s_{i-1}\in st_{i-1}$ and $\alpha_i\in h_A^{-1}(\tilde{\alpha}_i)$ such that $s_{i-1}\xrightarrow{\alpha_i}s_i$. Hence, there exists a concrete path from $s_{i-1}$ to $s_{i+k}$.
	
	\item Consider that $\tilde{\pi}'=\tilde{s}_{i-1}\sim_a\tilde{s}_{i}\,||\,\tilde{\pi}$ is a path of length $k+1$ where  $(\tilde{s}_{i-1}\sim_a \tilde{s}_{i}\,||\,\tilde{\pi},st_{i-1})\Rightarrow_e (\tilde{\pi},st_i) \Rightarrow^* (\tilde{s}_{i+k},st_{i+k})$. By induction hypothesis, there exists a concrete path that begins at some state $s_i\in st_i$ and ends in $s_{i+k}\in st_{i+k}$. By the definition of $\Rightarrow_e$ and proposition \ref{prop:temp-path}, $s_i$ is reachable from some initial states in the concrete model, which is a requirement by definition \ref{def:corr-path}. From $l_a(s_i)\in L_a(st_{i-1})$ we conclude that there exists $s_{i-1}\in st_{i-1}$ such that $l_a(s_i)=l_a(s_{i-1})$. Hence we select $s_{i-1}\sim_a s_{i}$ as the corresponding epistemic transition in the concrete model. Therefore, there exists a concrete path from $s_{i-1}$ to $s_{i+k}$.
\end{itemize}
\end{proof}

In the case that $\tilde{\pi}=\tilde{s}_0\rightarrow\dots\rightarrow\tilde{s}_n$ is a path in the counterexample and $(\tilde{\pi},S_0\cap h^{-1}(\tilde{s}_0))\Rightarrow^* (\tilde{s}_n,st_n)$, then there exists a corresponding concrete path beginning at some initial state $s_0\in S_0\cap h^{-1}(\tilde{s}_0)$ which ends at some state $s_n\in st_n$.

\begin{proposition}[Completeness of $\Rightarrow^*$]
	Let $\tilde{\pi}=\tilde{s}_1\rightarrow\dots\rightarrow\tilde{s}_n$ be a path in the abstract model. If there exists a concrete path $\pi=s_1\rightarrow\dots\rightarrow s_n$ corresponding to $\tilde{\pi}$ and $s_1\in st_1\subseteq h^{-1}(\tilde{s}_1)$, then $(\tilde{\pi},st_1)\Rightarrow^* (\tilde{s}_n,st_n)$ for some $\emptyset \subset st_n \subseteq S$.
\end{proposition}

\begin{proof}
	For the completeness proof, we use induction over the length of the counterexamples. 

\textbf{Base case:} $\tilde{\pi}=\tilde{s}_1$ and $\pi=s_1$. Then we will have no transition and the proposition automatically holds.

\textbf{Inductive case:} Assume by inductive hypothesis that for all $\tilde{\pi}=\tilde{s}_i\rightarrow\dots\rightarrow\tilde{s}_{i+k}$ of length $k$, if there exists a path $\pi=s_i\rightarrow\dots\rightarrow s_{i+k}$ which corresponds to $\tilde{\pi}$ and $s_i\in st_i\subseteq h^{-1}(\tilde{s}_i)$, then $(\tilde{\pi},st_i)\Rightarrow^* (\tilde{s}_{i+k},st_{i+k})$ for some $\emptyset \subset st_{i+k} \subseteq S$.

\begin{itemize}
	\item Consider that $\tilde{s}_{i-1}\xrightarrow{\tilde{\alpha}_i}\tilde{s}_{i}\,||\,\tilde{\pi}$ is a path of length $k+1$ which has the corresponding concrete path $s_{i-1}\xrightarrow{\alpha_i}s_{i}\,||\,\pi$. Let $st_{i-1}\in h^{-1}(\tilde{s}_{i-1})$ be a set of states where $s_{i-1}\in st_{i-1}$. Then the transition $(\tilde{s}_{i-1}\xrightarrow{\tilde{\alpha}_i}\tilde{s}_{i}\,||\,\tilde{\pi},st_{i-1})\Rightarrow_t (\tilde{\pi},st_i)$ leads to the set $st_i$ as the successors of the states in $st_{i-1}$ with respect to the actions in $h_A^{-1}(\tilde{\alpha}_{i})$. As $\alpha_i\in h_A^{-1}(\tilde{\alpha}_{i})$, we have $s_i \in st_i$. Therefore by inductive hypothesis,  we have $(\tilde{s}_{i-1}\xrightarrow{\tilde{\alpha}_i}\tilde{s}_{i}\,||\,\tilde{\pi},st_{i-1})\Rightarrow_t (\tilde{\pi},st_i)\Rightarrow^* (\tilde{s}_{i+k},st_{i+k})$ or equivalently $(\tilde{s}_{i-1}\xrightarrow{\tilde{\alpha}_i}\tilde{s}_{i}\,||\,\tilde{\pi},st_{i-1})\Rightarrow^* (\tilde{s}_{i+k},st_{i+k})$.

	\item Consider that $\tilde{s}_{i-1}\sim_a\tilde{s}_{i}\,||\,\tilde{\pi}$  is a path of length $k+1$ which has the corresponding concrete path $s_{i-1}\sim_a s_{i}\,||\,\pi$. Let $st_{i-1}\in h^{-1}(\tilde{s}_{i-1})$ be a set of states where $s_{i-1}\in st_{i-1}$. Then the transition $(\tilde{s}_{i-1}\sim_a\tilde{s}_{i}\,||\,\tilde{\pi},st_{i-1})\Rightarrow_e (\tilde{\pi},st_i)$ leads to the set $st_i$ which contains the reachable states with the same local states as the states in $st_{i-1}$. Therefore, $s_i\in st_i$ and by inductive hypothesis we have $(\tilde{s}_{i-1}\sim_a \tilde{s}_{i}\,||\,\tilde{\pi},st_{i-1})\Rightarrow_e (\tilde{\pi},st_i)\Rightarrow^* (\tilde{s}_{i+k},st_{i+k})$ or equivalently $(\tilde{s}_{i-1}\sim_a \tilde{s}_{i}\,||\,\tilde{\pi},st_{i-1})\Rightarrow^* (\tilde{s}_{i+k},st_{i+k})$.
\end{itemize}

\end{proof}

Forward transition rules in figure \ref{fig:forwardC} are sufficient to check \emph{linear counterexamples} or equivalently, paths. To extend the counterexample checking to tree-like counterexample, extra procedures are required. We show the problem in the following example:

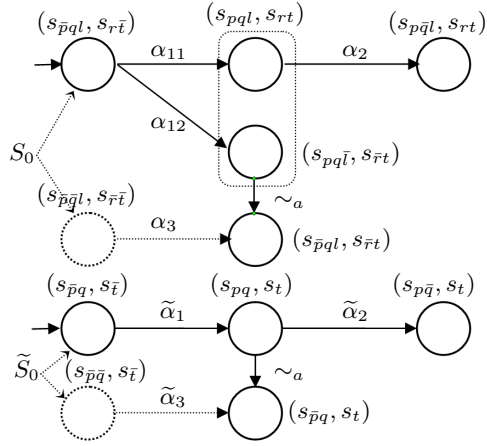
\begin{figure}[t]
	\centering
	\begin{overpic}[width=6cm]{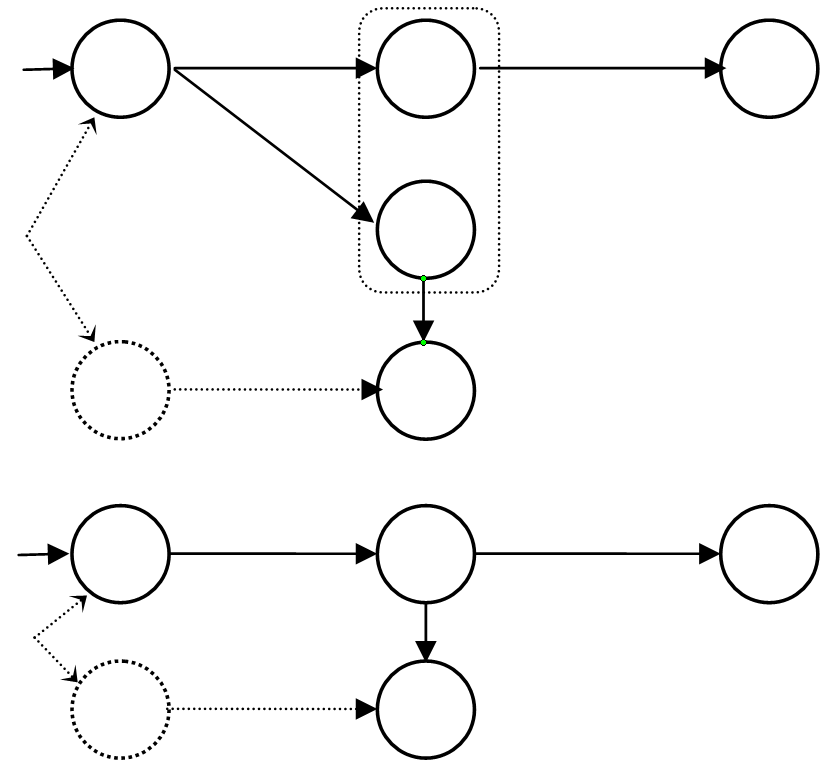}
		\put(3,92){$(s_{\bar{p}ql},s_{r\bar{t}})$}
		\put(40,94){$(s_{pql},s_{rt})$}		
		\put(80,92){$(s_{p\bar{q}l},s_{rt})$}		
		\put(61,63){$(s_{pq\bar{l}},s_{\bar{r}t})$}		
		\put(3,54){$(s_{\bar{p}\bar{q}l},s_{\bar{r}\bar{t}})$}		
		\put(59,44){$(s_{\bar{p}ql},s_{\bar{r}t})$}		
		\put(-3,63){$S_0$}		
		\put(28,86){$\alpha_{11}$}
		\put(28,70){$\alpha_{12}$}
		\put(70,86){$\alpha_2$}
		\put(28,48){$\alpha_3$}
		\put(55,53){$\sim_a$}
		\put(5,34){$(s_{\bar{p}q},s_{\bar{t}})$}
		\put(41,34){$(s_{pq},s_{t})$}		
		\put(81,34){$(s_{p\bar{q}},s_{t})$}	
		\put(9,15){$(s_{\bar{p}\bar{q}},s_{\bar{t}})$}		
		\put(58,6){$(s_{\bar{p}q},s_{t})$}		
		\put(-2,15){$\tilde{S}_0$}
		\put(30,28){$\tilde{\alpha}_1$}
		\put(70,28){$\tilde{\alpha}_2$}
		\put(30,9){$\tilde{\alpha}_3$}
		\put(55,15){$\sim_a$}
	\end{overpic}
	\caption{The transition system on the top is the concrete model and on the bottom is the abstract one obtained by making the propositions $l$ and $r$ invisible.}
	\label{fig:example}
\end{figure}

\begin{example}
\label{ex:spurious}
Figure \ref{fig:example} demonstrates the transition system for a concrete interpreted system on top, and the abstract system on the bottom. The model contains two agents, $e$ as the environment and $a$ as regular agent. States are shown as tuples where the first element is the local state of $e$ and the second is the local state of $a$. The diagram distinguishes the states by using the value of local propositions as the subscript. The abstract model is generated by making the local proposition $l$ of environment and $r$ of agent $a$ invisible.

We aim to verify $AG(p\rightarrow (K_a p \vee AGq))$ over the concrete model. This property holds for the original model, while it does not hold for the abstract one. The counterexample generated is:
\begin{equation*}
	\tilde{ce}=\{(s_{\bar{p}q},s_{\bar{t}})\xrightarrow{\tilde{\alpha}_1} (s_{pq},s_{t})\xrightarrow{\tilde{\alpha}_2}(s_{p\bar{q}},s_{t}),
		(s_{\bar{p}q},s_{\bar{t}})\xrightarrow{\tilde{\alpha}_1} (s_{pq},s_{t})\sim_a (s_{\bar{p}q},s_{t})
		\}
\end{equation*}

To find out if there exists any concrete counterexample that corresponds to $\tilde{ce}$, we check the paths in $\tilde{ce}$ one by one. We show the paths in $\tilde{ce}$ by $\tilde{\pi}_1$ and $\tilde{\pi}_2$. The paths $\tilde{\pi}_1$ and $\tilde{\pi}_2$ correspond to the concrete paths $\pi_1=(s_{\bar{p}ql},s_{r\bar{t}})\xrightarrow{\alpha_{11}} (s_{pql},s_{rt})\xrightarrow{\alpha_2}(s_{p\bar{q}l},s_{rt})$ and $\pi_2=(s_{\bar{p}ql},s_{r\bar{t}})\xrightarrow{\alpha_{12}} (s_{pq\bar{l}},s_{\bar{r}t})\sim_a (s_{\bar{p}ql},s_{\bar{r}t})$. Although all the paths in the counterexample have corresponding concrete paths, the tree does not correspond to a concrete tree. This is because if we select $(s_{pql},s_{rt})$ as the corresponding state for $(s_{pq},s_{t})$, then the leaf $(s_{\bar{p}ql},s_{\bar{r}t})$ is not reachable from it. A similar situation happens when we select $(s_{pq\bar{l}},s_{\bar{r}t})$. Therefore, the tree-like counterexample is spurious.
\end{example}

\begin{figure*}[t]
\begin{mathpar}
	\inferrule* [left=BackwardTCheck]
		{(\pi,S_0\cap h^{-1}(\textbf{Root}(\pi)))\Rightarrow^* (\tilde{s},st') \\
                          h_A^{-1}(\tilde{\alpha})=\{\alpha_1,\dots ,\alpha_n\} \\ 
		  rs=\bigcup_{i=1}^{n}\Theta^{-1}_{\alpha_i}(st)\cap st'}
		{(\pi\,||\,\tilde{s}\xrightarrow{\tilde{\alpha}}\tilde{s}',st) \Leftarrow_t (\pi, rs) \\ r_{\tilde{s}} :=rs}
\end{mathpar}
\begin{mathpar}
	\inferrule* [left=BackwardECheck]
		{
(\pi,S_0\cap h^{-1}(\textbf{Root}(\pi)))\Rightarrow^* (\tilde{s},st'') \\
\pi'=\tilde{s}'_0\xrightarrow{\tilde{\alpha}'_1}\dots\xrightarrow{\tilde{\alpha}'_m}\tilde{s}' \text{ is the temporal path to }\tilde{s}'\text{ where } \tilde{s}'_0\in \tilde{S}_0\\
		(\pi',S_0\cap h^{-1}(\tilde{s}'_0))\Rightarrow^* (\tilde{s}',st')\\
		\hat{st}=\{s\in st'' ~|~l_a(s)\in L_a(st\cap st')\}\\
}
		{(\pi\,||\,\tilde{s}\sim_a\tilde{s}',st) \Leftarrow_e (\pi,\hat{st}) \\ r_{\tilde{s}}:=\hat{st}} 
\end{mathpar}
	\caption{Backward temporal and epistemic transition traversal. $\Theta^{-1}_{\alpha}(st)$ computes the set of predecessors of the states in $st$ with respect to the transitions made by action $\alpha$.}
	\label{fig:backwardC}
\end{figure*}

To verify a tree-like counterexample, we introduce two transition rules \RefTirName{BackwardTCheck} and \RefTirName{BackwardECheck} denoted by $\Leftarrow_t$ and $\Leftarrow_e$. The transition rules find all the predecessors of the states in $st$ (figure \ref{fig:backwardC}) with respect to the temporal or epistemic transitions in a backward manner which reside in the set of reachable states through the path. We write $\Leftarrow^*$ to denote a sequence of backward transitions $\Leftarrow_t$ and $\Leftarrow_e$.

Assume that $\tilde{\pi}=\tilde{s}_0\rightarrow\dots\rightarrow\tilde{s}_n$ is a path in the counterexample $\tilde{ce}$ which $(\tilde{\pi},S_0\cap h^{-1}(\tilde{s}_0))\Rightarrow^* (\tilde{s}_n,st_n)$ for some $\emptyset\subset st_n\subseteq S$. $st_n$ contains all the states in the leaves of the concrete paths corresponding to $\tilde{\pi}$. The point is not all the concrete states that are traveresed in $\Rightarrow^*$ can reach the states in $st_n$. If $\tilde{s}\in\textbf{Vert}(\tilde{\pi})$, then $(\tilde{\pi},st_n)\Leftarrow^* (\tilde{s}_0,st_0)$ finds the set of states $r_{\tilde{s}}$ which contains the reachable states in $h^{-1}(\tilde{s})$ that lead to some states in $st_n$ along the concrete paths corresponding to $\tilde{\pi}$. $st_0$ contains the initial states that lead to the states in $st_n$. We use the notation $r^{\tilde{\pi}}_{\tilde{s}}$ to relate $r_{\tilde{s}}$ with the path $\tilde{\pi}$. Note that to find $r^{\tilde{\pi}}_{\tilde{s}}$, we first need to find $st_n$ through $\Rightarrow^*$ transition.

Assume that $\tilde{\Pi}\subseteq \tilde{ce}$. If $\tilde{s}\in\textbf{Vert}(\tilde{ce})$  then we define $r_{\tilde{s}}^{\tilde{\Pi}}=\cap_{\tilde{\pi}\in\tilde{\Pi}} r_{\tilde{s}}^{\tilde{\pi}}$. If $\tilde{s}\not\in\textbf{Vert}(\tilde{\pi})$, then we stipulate $r_{\tilde{s}}^{\tilde{\pi}}=h^{-1}(\tilde{s})$. We also stipulate $r_{\tilde{s}_0}^{\emptyset}=S_0\cap h^{-1}(\tilde{s}_0)$ where $\tilde{s}_0=\textbf{Root}(\tilde{ce})$ and $r_{\tilde{s}}^{\emptyset}=h^{-1}(\tilde{s})$ for all $\tilde{s}\in\textbf{Vert}(\tilde{ce})$ where $\tilde{s}\neq \tilde{s}_0$. 

\begin{proposition}[Soundness of counterexample checking]
\label{prop:sound}
A counterexample $\tilde{ce}$ in the abstract model has a corresponding concrete one if:
	\begin{enumerate}
		\item \label{item:A} for each path $\tilde{\pi}\in \tilde{ce}$, there exists $\emptyset\subset st \subseteq S$ such that $(\tilde{\pi},S_0\cap h^{-1}(\tilde{s}_0))\Rightarrow^* (\tilde{s}',st)$ where $\tilde{s}_0=\textbf{Root}(\tilde{ce})$ and $\tilde{\pi}$ ends in $\tilde{s}'$.
		\item \label{item:B} for all $\tilde{s}\in\textbf{Vert}(\tilde{ce}): r^{\tilde{ce}}_{\tilde{s}}\neq\emptyset$.
	\end{enumerate}
\end{proposition}

\begin{proof}
By the soundness of $\Rightarrow^*$, all the paths in $\tilde{\pi}$ correspond to some concrete paths which satisfy the requirements in the definitions \ref{def:corr-path} and \ref{def:conc-ce}. Now for each $\tilde{s}\in \textbf{Vert}(\tilde{ce})$, we pick a state $s\in r^{\tilde{ce}}_{\tilde{s}}$ as the corresponding state. For each path in $\tilde{ce}$ and between all the corresponding concrete paths, we pick the one which contains the selected states as its vertices. The union of the selected paths builds a concrete counterexample that satisfies the requirements in definition \ref{def:conc-ce}.
\end{proof}

\begin{algorithm*}
\floatname{algorithm}{Procedure}
\caption{Counterexample checking algorithm}\label{ceCheck}
\begin{algorithmic}
\Function{CheckCE}{$\tilde{ce},I,h$} 
\State $\triangleright$ \textbf{Input}: $\tilde{ce}$ is the counterexample, $I$ is the concrete model and $h$ is the abstraction function
\State $\triangleright$ \textbf{Output}: returns \texta{true} if a concrete counterexample exists. Returns \texta{false} otherwise.
	\State $\{\tilde{s}_0,\dots,\tilde{s}_n\}=\textbf{Vert}(\tilde{ce})$ \Comment{$\tilde{s}_0=\textbf{Root}(\tilde{ce})$}
	\State $\tilde{\Pi}=\emptyset$
	\State $r^{\tilde{\Pi}}_{\tilde{s_0}}=S_0\cap h^{-1}(\tilde{s}_0), r^{\tilde{\Pi}}_{\tilde{s}_1}=h^{-1}(\tilde{s}_1),\dots,r^{\tilde{\Pi}}_{\tilde{s}_n}=h^{-1}(\tilde{s}_n)$
	\ForAll {$\tilde{\pi}\in \tilde{ce}$}
		\If {$(\tilde{\pi},r^{\tilde{\Pi}}_{\tilde{s_0}})\Rightarrow^* (\tilde{s}',st)$ and $st\neq\emptyset$} \Comment{$\tilde{\pi}$ ends at the state $\tilde{s}'$}
		\State $\triangleright$ there exists some concrete path corresponding to $\tilde{\pi}$
			\ForAll {$\tilde{s}\in\textbf{Vert}(\tilde{ce})$}
				\State \textbf{determine } $\hat{r}_{\tilde{s}}^{\tilde{\pi}}$ \textbf{ from } $(\tilde{\pi},st)\Leftarrow^* (\tilde{s}_0,st')$
				\State $\triangleright$ determine the concrete states corresponding to $\tilde{s}$
				\State $r^{\tilde{\Pi}\cup\{\tilde{\pi}\}}_{\tilde{s}}:=r^{\tilde{\Pi}}_{\tilde{s}}\cap r_{\tilde{s}}^{\tilde{\pi}}$
				\If {$r^{\tilde{\Pi}\cup\{\tilde{\pi}\}}_{\tilde{s}}=\emptyset$} 
				\State $\triangleright$ no common concrete state for $\tilde{s}$ between concrete paths exists
					\State \textbf{return }\texta{false}
				\EndIf
			\EndFor
			\State $\tilde{\Pi}:=\tilde{\Pi}\cup\{\tilde{\pi}\}$
		\Else
			\State \textbf{return }\texta{false}
		\EndIf
	\EndFor
	\State \textbf{return }\texta{true}
\EndFunction
\end{algorithmic}
\label{algo:cecheck}
\end{algorithm*}

\begin{proposition}[Completeness of counterexample checking]
	Assume that $\tilde{ce}$ corresponds to a concrete counterexample $ce$. Then both the items \ref{item:A} and \ref{item:B} in proposition \ref{prop:sound} hold.
\end{proposition}

\begin{proof}
By definition \ref{def:conc-ce}, there is a one-to-one correspondence between the paths of the two counterexamples. By completeness of $\Rightarrow^*$, item \ref{item:A} holds for all the paths in $\tilde{ce}$. Now Assume that $\tilde{s}\in \textbf{Vert}(\tilde{ce})$ and $s$ is the corresponding state in $ce$. Then for all $\tilde{\pi}\in \tilde{ce}$, we have $s\in r^{\tilde{\pi}}_{\tilde{s}}$, and therefore $s\in r^{\tilde{ce}}_{\tilde{s}}$. Hence we have $r^{\tilde{ce}}_{\tilde{s}}\neq\emptyset$, as required for item \ref{item:B}.
\end{proof}

Procedure \ref{algo:cecheck} expresses the tree-like counterexample checking method in a more refined manner.
\textsc{CheckCE} iterates over the paths in $\tilde{ce}$ and checks if they corresponds to some paths in the concrete model by using proposition \ref{prop:soundnessoftemp} and the transition rule $\Rightarrow^*$. If $\tilde{\pi}$ corresponds to some concrete paths, then for each state $\tilde{s}$ in $\tilde{\pi}$, the algorithm finds all the concrete states $r_{\tilde{s}}^{\tilde{\pi}}$ in $h^{-1}(\tilde{s})$ that lead to the leaf states of the concrete paths by applying $\Leftarrow^*$ over $\tilde{\pi}$. In each loop iteration, $\tilde{\Pi}$ stores the paths in $\tilde{ce}$ that are processed in previous iterations. The set $r^{\tilde{\Pi}}_{\tilde{s}}$ stores the concrete states that are common between the paths in $\tilde{\Pi}$ and should remain non-empty during the process of counterexample checking. The procedure returns \texta{false} if no corresponding tree-like counterexample for $\tilde{ce}$ exists. Otherwise it returns \texta{true}.

\begin{example}
We recall the transition system in example \ref{ex:spurious}. As also discovered in the example, the paths $\tilde{\pi}_1$ and $\tilde{\pi}_2$ correspond to the concrete paths $\pi_1=(s_{\bar{p}ql},s_{r\bar{t}})\xrightarrow{\alpha_{11}} (s_{pql},s_{rt})\xrightarrow{\alpha_2}(s_{p\bar{q}l},s_{rt})$ and $\pi_2=(s_{\bar{p}ql},s_{r\bar{t}})\xrightarrow{\alpha_{12}} (s_{pq\bar{l}},s_{\bar{r}t})\sim_a (s_{\bar{p}ql},s_{\bar{r}t})$. By backward traversing through the first path and for the states in $h^{-1}((s_{pq},s_{t}))$, we find that only the state $(s_{pql},s_{rt})$ leads to the final state on $\pi_1$ and so, $r_{(s_{pq},s_{t})}^{\tilde{\pi}_1}=\{(s_{pql},s_{rt})\}$. The same approach for $\pi_2$ results in $r_{(s_{pq},s_{t})}^{\tilde{\pi}_2}=\{(s_{pq\bar{l}},s_{\bar{r}t})\}$. As $r_{(s_{pq},s_{t})}^{\tilde{\pi}_1}\cap r_{(s_{pq},s_{t})}^{\tilde{\pi}_2}=\emptyset$, the state $(s_{pq},s_{t})$ can not be assigned to a concrete single state. Therefore, $\tilde{ce}$ is spurious.
\end{example}

\subsection{Refinement of the abstraction}
\label{sec:refinement}

If the counterexample is found to be spurious, then the abstraction should be refined. The abstract model is generated by making some propositions in the concrete model invisible. For the refinement, we split some states in the abstract model by putting some of the invisible propositions back into the model. These propositions should be selected in such a way that when verifying the refined model, the same counterexample does not appear again. In this section, we provide the mechanism for refining the abstraction.

Let $\tilde{ce}$ be a spurious counterexample. We define two transition rules \RefTirName{TemporalTree} which is denoted by $\Rightarrow^{\tilde{\Pi}}_t$ and \RefTirName{EpistemicTree} denoted by $\Rightarrow^{\tilde{\Pi}}_e$ where $\tilde{\Pi}\subseteq\tilde{ce}$ in figure \ref{fig:failureC}. As before, $\Rightarrow^{\tilde{\Pi}}_*$ denotes a sequence of temporal and epistemic transitions of the type $\Rightarrow^{\tilde{\Pi}}_t$ and $\Rightarrow^{\tilde{\Pi}}_e$. We use the following technique in order to find the state in the spurious counterexample which needs to be split:
\vspace{2mm}

The state $\tilde{s}_i\in\textbf{Vert}(\tilde{ce})$ is a \emph{failure state} if there exists $\tilde{\Pi}\subseteq\tilde{ce}$ and $\tilde{\pi}\in \tilde{ce}\backslash\tilde{\Pi}$ such that:
\begin{enumerate}
	\item For all $\tilde{s}\in\textbf{Vert}(\tilde{\Pi}):r_{\tilde{s}}^{\tilde{\Pi}}\neq\emptyset$
	\item $\tilde{\pi}=\tilde{\pi}_1~||~\tilde{s}_i(\xrightarrow{\tilde{\alpha}_{i+1}}|\sim_a)\tilde{s}_{i+1}~||~\tilde{\pi}_2$ such that $(\tilde{\pi},r_{\tilde{s}_0}^{\tilde{\Pi}})\Rightarrow^{\tilde{\Pi}}_* (\tilde{\pi}_1,st_d)\Rightarrow^{\tilde{\Pi}}_{(t|e)} (\tilde{\pi}_2,\emptyset)$ for some $st_d\neq \emptyset$.
\end{enumerate}

\begin{figure*}[t]
\begin{mathpar}
	\inferrule* [left=TemporalTree]
		{h_A^{-1}(\tilde{\alpha})=\{\alpha_1,\dots ,\alpha_n\}}
		{(\tilde{s}\xrightarrow{\tilde{\alpha}}\tilde{s}'\,||\,\pi,st) \Rightarrow^{\tilde{\Pi}}_t (\pi,\bigcup_{i=1}^{n}\Theta_{\alpha_i}(st)\cap r_{\tilde{s}'}^{\tilde{\Pi}})}
\end{mathpar}
\begin{mathpar}
	\inferrule* [left=EpistemicTree]
		{\pi'=\tilde{s}'_0\xrightarrow{\tilde{\alpha}'_1}\dots\xrightarrow{\tilde{\alpha}'_m}\tilde{s}' \text{ is a temporal path to }\tilde{s}'\text{ where } \tilde{s}'_0\in \tilde{S}_0\\
		(\pi',S_0\cap h^{-1}(\tilde{s}'_0))\Rightarrow^*_t (\tilde{s}',st')\\
		\hat{st}=\{s\in st'\cap r_{\tilde{s}'}^{\tilde{\Pi}}~|~l_a(s)\in L_a(st)\}
}
		{(\tilde{s}\sim_a\tilde{s}'\,||\,\pi,st) \Rightarrow^{\tilde{\Pi}}_e (\pi,\hat{st})}
\end{mathpar}
	\caption{Transition rules for finding failure state in a tree-like counterexample.}
	\label{fig:failureC}
\end{figure*}

For a spurious counterexample, such $\tilde{\Pi}$ and $\tilde{\pi}$ exists. Otherwise, we will have $r_{\tilde{s}}^{\tilde{ce}}\neq\emptyset$ for all $\tilde{s}\in\textbf{Vert}(\tilde{ce})$, which contradicts proposition \ref{prop:sound}.

Based on Item 1), the sub-tree $\tilde{\Pi}$ has a corresponding counterexample in the concrete model. In item 2), $\tilde{\pi}$ traverses over the concrete states that belong to the set of concrete trees corresponding to $\tilde{\Pi}$ and gets to the set of states $st_d\subseteq h^{-1}(\tilde{s_i})$ with no transition to a state in $r^{\tilde{\Pi}}_{\tilde{s}_{i+1}}$. In the standard terminology as in \cite{CEGAR}, $\tilde{s}_i$ is called \emph{failure state}. We use the term \emph{dead end state} for the states in $st_d$ which the concrete paths end up with and can not go further. \emph{Bad states} are the states in $h^{-1}(\tilde{s}_i)$ that have transition to some states in $r^{\tilde{\Pi}}_{\tilde{s}_{i+1}}$. Note that in a path counterexample, we have that $r^{\tilde{\Pi}}_{\tilde{s}_{i+1}}=h^{-1}(\tilde{s}_{i+1})$. 

The process of finding a failure state in the counterexample $\tilde{ce}$ proceeds as follows:

\begin{enumerate}
	\item Set $\tilde{\Pi}$ to empty set at the beginning
	\item Find $r_{\tilde{s}}^{\tilde{\Pi}}$ for all $\tilde{s}\in\textbf{Vert}(\tilde{ce})$ (as also mentioned in section \ref{sec:validation}, $r_{\tilde{s}_0}^{\emptyset}=S_0\cap h^{-1}(\tilde{s}_0)$ where $\tilde{s}_0=\textbf{Root}(\tilde{ce})$ and $r_{\tilde{s}}^{\emptyset}=h^{-1}(\tilde{s})$ for all $\tilde{s}\in\textbf{Vert}(\tilde{ce})$ where $\tilde{s}\neq \tilde{s}_0$)
	\item Pick a path $\tilde{\pi}\in\tilde{ce}$ that does not exist in $\tilde{\Pi}$
	\item Apply $\Rightarrow^{\tilde{\Pi}}_*$ over $(\tilde{\pi},r_{\tilde{s}_0}^{\tilde{\Pi}})$ to find failure state. If a failure state exists over $\tilde{\pi}$, then exit and refine the model
	\item Add $\tilde{\pi}$ to $\tilde{\Pi}$ and return to step 2. Note that we are considering that the counterexample is found to be spurious (by the procedure \ref{algo:cecheck}) and therefore, such failure state will be found before all the paths in $\tilde{ce}$ are added to $\tilde{\Pi}$.
\end{enumerate}

For the implementation, the above process can be easily incorporated into the procedure \ref{algo:cecheck}.

To refine the model, we find the propositions that having them invisible results in generating spurious counterexample. First assume that the transition from $\tilde{s}_i$ to $\tilde{s}_{i+1}$ is temporal, say $\tilde{s}_i\xrightarrow{\tilde{\alpha}_{i+1}}\tilde{s}_{i+1}$. Two situations can result in a transition of type $\Rightarrow^{\tilde{\Pi}}_t$ from $st_d$ to an empty set of states: 
\begin{itemize}
	\item There exists no $\alpha_{i+1}\in h^{-1}(\tilde{\alpha}_{i+1})$ such that $\Theta_{\alpha_{i+1}}(st_d)\neq\emptyset$. Therefore, no action has the permission to be performed on the states of $st_d$. Assume that $\phi_d$ is the formula that represents the set of states $st_d$. As the state space is finite, the formula representing the states always exists. Therefore, for all $\alpha_{i+1}\in h_A^{-1}(\tilde{\alpha}_{i+1})$ with $\ell_{i+1}$ as the permission, we have $\phi_d\wedge \ell_{i+1}\equiv\bot$. We call $\ell_{i+1}$ \emph{conflict formula} and $\phi_d$ \emph{base formula}. 

	\item For some $\alpha_{i+1}\in h^{-1}(\tilde{\alpha}_{i+1})$ we have $\Theta_{\alpha_{i+1}}(st_d)\neq\emptyset$. By the definition of $\Rightarrow_t$ we have $\Theta_{\alpha_{i+1}}(st_d)\cap  r_{\tilde{s}'_{i+1}}^{\tilde{\Pi}}=\emptyset$ where $r_{\tilde{s}'_{i+1}}^{\tilde{\Pi}}\neq\emptyset$. If $\phi$ is the formula representing $\Theta_{\alpha_{i+1}}(st_d)$ and $\psi$ the formula representing $r_{\tilde{s}'_{i+1}}^{\tilde{\Pi}}$, then we have $\psi\wedge \phi\equiv\bot$. We call $\phi$ conflict formula and $\psi$ base formula.
\end{itemize}

The other situation is when the transition $\tilde{s}_i$ and $\tilde{s}_{i+1}$ is epistemic, say $\tilde{s}_i\sim_a\tilde{s}_{i+1}$. Three situations can result in the epistemic transition $\Rightarrow^{\tilde{\Pi}}_e$ to an empty set of states:
\begin{itemize}
	\item $\pi'$ as the witness of the reachability of $\tilde{s}_{i+1}$ in $\Rightarrow^{\tilde{\Pi}}_e$ is spurious. Then the refinement should be guided by analysing $\pi'$ instead of the main spurious path.

	\item Suppose that $\pi'$ has corresponding concrete paths, i.e. $(\pi',S_0\cap h^{-1}(\tilde{s}'_0))\Rightarrow^*_t (\tilde{s}_{i+1},st')$ where $st'\neq\emptyset$. By the definition of $\Rightarrow_e$, the epistemic transition results in an empty set of states if $st'\cap r_{\tilde{s}'_{i+1}}^{\tilde{\Pi}}=\emptyset$. If $\phi$ is the formula representing $st'$ and $\psi$ the formula representing $r_{\tilde{s}'_{i+1}}^{\tilde{\Pi}}$, then we call $\phi$ conflict formula and $\psi$ base formula.

	\item The third reason for the epistemic transition to an empty set is when no shared \emph{local state} exists  between the states of $st_d$ and $st'\cap r_{\tilde{s}'_{i+1}}^{\tilde{\Pi}}$ where $st'$ is the set of reachable states according to the previous item and both the sets are non-empty. In the other words, $L_a(st_d)\cap L_a(st'\cap r_{\tilde{s}'_{i+1}}^{\tilde{\Pi}})=\emptyset$. The formula representing the local states in $st_d$ with respect to the agent $a$ is called base formula, and the formula representing the local states of $st'\cap r_{\tilde{s}'_{i+1}}^{\tilde{\Pi}}$ is the conflict formula.
\end{itemize}

To refine the model, we return some hidden propositions to separate the set of dead end states from the rest of the states. This can simply be done by adding all the propositions occurring in \emph{conflict clauses} to the abstract model.

\begin{definition}(conflict clause)
Let $\phi$ be the base formula and $\psi$ the conflict formula. Let \textbf{cnf}($\psi$) denote the set containing all the conjuncts appear in conjunctive normal form of $\psi$. Then $c\in \text{\textbf{cnf}}(\psi)$ is a \emph{conflict clause} if $c\wedge\phi\equiv\bot$.
\end{definition}

If the propositions that occur in one of the conflict clauses become visible, then the spurious strategy will not happen in the refined model again. In the case of temporal transition, we add the propositions in the conflict clauses for \emph{all the conflicting actions}. To have the smallest possible refinement, we should look for the conflict classes with the \emph{smallest number} of literals.

\subsection{Going beyond ACTLK}
\label{sec:otherprops}

While this section develops a fully automated abstraction refinement method for the verification of temporal-epistemic properties that reside the category of ACTLK over an access control system which is modelled by an interpreted system, some important epistemic safety properties does not reside in this category. For instance and in a conference paper review system, it is valuable for policy designers to verify that for all reachable states, an author of a paper cannot find out ($\neg K$) who is the reviewer of his own paper (see the first property in example \ref{ex:ctlk}). Although we are able to verify such properties in the concrete model, we cannot apply automated counterexample-guided abstraction and refinement for such properties.

Let us explore the problem. Assume that for the abstract system $\tilde{I}$, abstract state $\tilde{s}$ and agent $a$, $(\tilde{I},\tilde{s})\models\neg K_a\varphi$. That means there exists a state $\tilde{s}'$ such that $\tilde{s}'\sim_a \tilde{s}$ and $(\tilde{I},\tilde{s}')\models\neg\varphi$. If $s$ is a state in the concrete model where $h(s)=\tilde{s}$, then  the satisfaction relation $(\tilde{I},\tilde{s})\models\neg K_a\varphi$ implies $(I,s)\models\neg K_a\varphi$ if it guarantees the existence of a \emph{reachable} state $s'\in h^{-1}(\tilde{s}')$ such that $s'\sim_a s$ and $(I,s')\models\neg\varphi$.

First of all, if such $s'$ exists, the satisfaction relation $(\tilde{I},\tilde{s}')\models\neg\varphi$ still does not imply $(I,s')\models\neg\varphi$ when $\varphi$ is ACTLK except if $\varphi$ is simply a propositional formula which is the case for many of the properties that we are interested in. Second, the relation $\tilde{s}'\sim_a \tilde{s}$ in the abstract model does not imply $s'\sim_a s$ in the concrete model for some reachable state $s'\in h^{-1}(\tilde{s}')$. In the case that $(\tilde{I},\tilde{s}')\not\models\neg\varphi$, the model-checker produces a counterexample that can be checked using the method that is developed in this section and then the abstract model can be refined. In the case that the satisfaction relation holds, the model-checker does not produce any witness.

To complete our work for the properties that deal with the negation of knowledge operator, we restrict the formula in scope of the knowledge operators to propositional formulas. Then we use an interactive refinement procedure in the following way: we abstract the interpreted system in the standard way that we described. If the property does not hold in the abstract model, the counterexample will be checked in the concrete model and the abstract model will be refined if it is required. If the property turned to be true in the abstract model as a result of the satisfaction of $\neg K_a$ (which we would not have any witness in the abstract model), then we refine the local state of the agent $a$ in an interactive manner. In this way, the tool asks the user to selects a set of invisible \emph{local propositions} to be added in the next round if required. This process will continue until a valid counterexample is found, or the local state becomes concretized. In the case that the safety property does not hold in the concrete model (where information leakage vulnerability exists), then there is a chance to find it out with the abstract model when the local states are still abstract.

\section{Experimental results}
\label{sec:experiment}

We have implemented a tool in F\# functional programming language. The font end is a parser that accepts a set of action and read permission rules, a set of objects and a query in the form of $\iota : \varphi$ where $\iota$ is the formula representing the initial states and $\varphi$ is the property we aim to verify. Given the above information, the tool derives an interpreted system based on definition~\ref{def:derivedis} where the initial states of the system are determined by parameter $\iota$ in the query. On the back end, we use MCMAS \cite{Lomuscio:2006:MCMAS} as the model-checking engine. In the presence of abstraction and refinement, the tool feeds MCMAS with the abstracted version of the original interpreted system together with the property $\varphi$. If model-checker returns \texta{true} for an ACTLK property, then the tool returns \texta{true} to the user. Otherwise, the tool automatically checks the generated counterexample based on proposition \ref{prop:sound}, and reports if it is a real counterexample, which will be returned to the user, or verification needs a refinement round. The tool performs an automated refinement if it is required. For the properties that are discussed in section~\ref{sec:otherprops}, the tool asks user to select a set of invisible local variables to be added to the abstract model for the refinement when model-checker returns \texta{true}. This will continue until all the related invisible local variables turn to visible, or a valid counterexample is found.

For this section, we choose one temporal and three epistemic properties for the case study of conference paper review system (CRS) with the information leakage vulnerability described in the introduction. We first verify the query (Query 1) ``$\texta{author}(\texta{p}_1, \texta{a}_1)\wedge \neg\texta{reviewer}(\texta{p}_1, \texta{a}_1) : AG(\neg\texta{reviewer}(\texta{p}_1, \texta{a}_1))$'' which states that if in the initial states, agent $\texta{a}_1$ is the author of paper $\texta{p}_1$ and not the reviewer of his own paper, then it is not possible for $\texta{a}_1$ to be assigned as the reviewer of his paper $\texta{p}_1$. Query 2 ``$\neg\texta{submittedreview}\\(\texta{p}_1,\texta{a}_1) \wedge \texta{reviewer}(\texta{p}_1,\texta{a}_2): AG(K_{\texta{a}_1}\texta{review}(\texta{p}_1,\texta{a}_2)\rightarrow AG(\neg \texta{submittedreview}(\texta{p}_1,\texta{a}_1))$ checks if in the initial states, $\texta{a}_1$ and $\texta{a}_2$ are the reviewers of paper $\texta{p}_1$, and $\texta{a}_1$ has not submitted his own review of $\texta{p}_1$, then $\texta{a}_1$ cannot submit her review if he \emph{reads} the review of $\texta{a}_2$ (knowledge by readability). Query 3 $\texta{author}(\texta{p}_1, \texta{a}_1) :AG(\texta{AllPapersAssigned}\wedge\texta{reviewer}(\texta{p}_1, \texta{a}_2)
\rightarrow \neg K_{\texta{a}_1}\texta{reviewer}(\texta{p}_1, \texta{a}_2))$ asks if $\texta{a}_1$ is the author of $\texta{p}_1$, then it is not possible for $\texta{a}_1$ to find the reviewer of his paper when his paper is assigned to $\texta{a}_2$, which is not ACTLK. Query 4 $\texta{author}(\texta{p}_1, \texta{a}_1) : AG(\texta{AllPapersAssigned}\wedge\texta{reviewer}(\texta{p}_1, \texta{a}_2)
\rightarrow K_{\texta{a}_1}\texta{reviewer}(\texta{p}_1, \texta{a}_2))$ has ACTLK property, which checks if $\texta{a}_1$ can always find who the reviewer of his paper is whenever all the papers are assigned.

\begin{SCfigure}[5][t]
	\centering
	\begin{overpic}[width=6cm,height=3.5 cm]{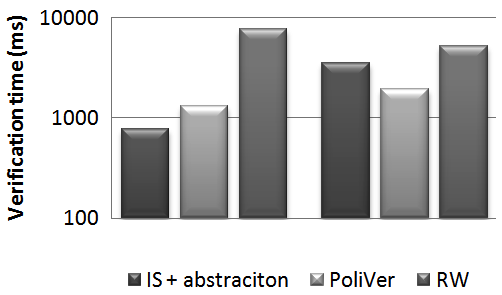}
		\put(27,58){\scriptsize 3 Papers, 7 Agents}
		\put(63,58){\scriptsize 2 Papers, 4 Agents}
		\put(33,10){\scriptsize Query 1}
		\put(74,10){\scriptsize Query 2}
	\end{overpic}
	\caption{Comparison of the verification time for the queries 1 and 2 between our tool which uses MCMAS as the model-checking engine, PoliVer and RW.}
\label{fig:toolCompare}
\end{SCfigure}
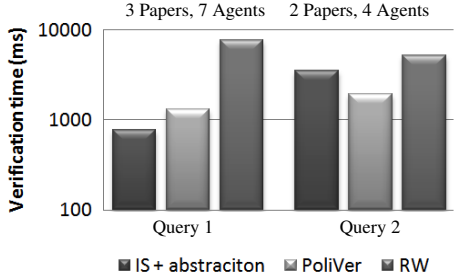

\begin{figure*}[t]
\center
{\begin{tabular}{| l | c | c | c | c | c | c |}
\hline
& \multicolumn{2}{| c |}{ Concrete model } & \multicolumn{4}{| c |}{ Abstraction and refinement }\\
\hline
	&time(s)&BDD vars&time(s)&Max BDD vars&last ref time&num of ref\\
\hline
Query 3&6576.5&180&148.3&80&3.28&7\\
\hline
Query 4&6546.4&180&174.1&98&21&12\\
\hline
\end{tabular}}
\caption{A comparison of query verification time (in second) and runtime memory usage (in MB) between the concrete model and automated abstraction refinement method.} 
\label{fig:table}
\end{figure*}

Queries 1 and 2 can be verified in access control policy verification tools like RW and PoliVer, which model knowledge by readability. We compare our tool in the presence of abstraction and refinement with RW and PoliVer from the point of verification time in figure \ref{fig:toolCompare}. It is important to note that when applying abstraction and refinement, a high percentage of evaluation time is spent on generating the whole concrete model at the beginning, invoking executable MCMAS which also invokes Cygwin library, generating abstract model and verifying the counterexample. In most of our experiments, verification of the final abstract model by MCMAS takes less than 10ms. 

The novel outcome of our research is the verification of the queries 3 and 4 where PoliVer and RW are unable to detect information leakage in CRS policy. In PoliVer and RW models, the author never finds a chance to see who the reviewer of his paper is and therefore safety property holds in the system. Modeling in interpreted systems reveals that the author can reason who is the reviewer of his paper when all the papers are assigned. For Query 3, the tool also outputs the counterexample which demonstrates the sequence of actions that allows the author to reason about the reviewer of his paper. Figure \ref{fig:table} shows the practical importance of our abstraction method (interactive refinement for Query 3 and fully automated for Query 4).

\section{Conclusion}

In this research, we introduced a framework for verifying temporal and epistemic properties over access control policies. In order to verify knowledge by reasoning, we used interpreted systems as the basic framework and to make the verification practical for medium to large systems, we extended counterexample-guided refinement known as CEGAR to cover safety properties in ACTLK. Case studies and experimental results show a considerable reduction in time and space when abstraction and refinement are in use. We also applied an interactive refinement for some useful properties that does not reside in ACTLK like the ones that contain the negation of knowledge modality.
As future work, we would like to use these technique to detect information-flow in real world systems such as electronic voting systems~\cite{Clarkson:2008:civitas,Singh:2011:Trivitas,Singh:2013:CC} and social networks.
\medskip

\noindent\textbf{Acknowledgement:} We would like to acknowledge Microsoft Research and EPSRC project TS/I002529/1 ``Trust Domains'' for funding this research.

\bibliographystyle{splncs}
\bibliography{References}

\begin{thebibliography}{10}

\bibitem{DynPAL}
Becker, M.Y.:
\newblock Specification and analysis of dynamic authorisation policies.
\newblock In: Proc. IEEE Computer Security Foundations Symposium. (July 2009)
  203--217

\bibitem{synthesising-nmd-2007}
Zhang, N., Ryan, M., Guelev, D.P.:
\newblock Synthesising verified access control systems through model checking.
\newblock Journal of Computer Security \textbf{16}(1) (2008)  1--61

\bibitem{DynACS-Dougherty}
Dougherty, D.J., Fisler, K., Krishnamurthi, S.:
\newblock Specifying and reasoning about dynamic access-control policies.
\newblock In: Proc. International Joint Conference on Automated Reasoning.
  (August 2006)  632--646

\bibitem{Mardare:2006:DynamicESLogics}
Mardare, R., Priami, C.:
\newblock Dynamic epistemic spatial logics.
\newblock Technical report, The Microsoft Research-University of Trento Centre
  for Computational and Systems Biology (2006)

\bibitem{Fagin-book-95}
Fagin, R., Halpern, J.Y., Moses, Y., Vardi, M.Y.:
\newblock Reasoning About Knowledge.
\newblock MIT Press, Cambridge (1995)

\bibitem{CEGAR}
Clarke, E.M., Grumberg, O., Jha, S., Lu, Y., Veith, H.:
\newblock Counterexample-guided abstraction refinement.
\newblock In: Proc. Computer Aided Verification. (July 2000)  154--169

\bibitem{Clarke-treelikecounterexamples}
Clarke, E.M., Lu, Y., Com, B., Veith, H., Jha, S.:
\newblock Tree-like counterexamples in model checking.
\newblock In: Proc. IEEE Symposium on Logic in Computer Science. (July 2002)
  19--29

\bibitem{Lomuscio:2006:MCMAS}
Lomuscio, A., Raimondi, F.:
\newblock {MCMAS}: A model checker for multi-agent systems.
\newblock In: proc. Tools and Algorithms for the Construction and Analysis of
  Systems. (April 2006)  450--454

\bibitem{Aucher2010-modal}
Aucher, G., Boella, G., van~der Torre, L.:
\newblock Privacy policies with modal logic: The dynamic turn.
\newblock In: Deontic Logic in Computer Science. (2010)  196--213

\bibitem{Koleini:2011:PoliVer}
Koleini, M., Ryan, M.:
\newblock A knowledge-based verification method for dynamic access control
  policies.
\newblock In: Proc. International Conference on Formal Engineering Methods.
  (2011)

\bibitem{Abstract-MAS-Cohen}
Cohen, M., Dam, M., Lomuscio, A., Russo, F.:
\newblock Abstraction in model checking multi-agent systems.
\newblock In: AAMAS 2009: Proceedings of The 8th International Conference on
  Autonomous Agents and Multiagent Systems. (2009)  945--952

\bibitem{Zhou-CEGAR}
Zhou, C., Sun, B., Liu, Z.:
\newblock Abstraction for model checking multi-agent systems.
\newblock Frontiers of Computer Science in China \textbf{5} (2011)  14--25

\bibitem{Fagin:97:knowledge}
Fagin, R., Halpern, J.Y., Moses, Y., Vardis, M.Y.:
\newblock Knowledge-based programs.
\newblock Distributed Computing \textbf{10}(4) (1997)  199--225

\bibitem{Lomuscio-mck-06}
Lomuscio, A., Raimondi, F.:
\newblock Model checking knowledge, strategies, and games in multi-agent
  systems.
\newblock In: Proc. International Conference on Autonomous Agents and
  Multiagent Systems. (May 2006)  161--168

\bibitem{ATL-Alure}
Alur, R., Henzinger, T.A., Kupferman, O.:
\newblock Alternating-time temporal logic.
\newblock Journal of the ACM \textbf{49}(5) (2002)  672--713

\bibitem{Lomuscio:2006:CMC}
Lomuscio, A., Raimondi, F.:
\newblock The complexity of model checking concurrent programs against {CTLK}
  specifications.
\newblock In: Proc. International Conference on Autonomous Agents and
  Multiagent Systems. (May 2006)  548--550

\bibitem{ClarkeMC}
Clarke, E.M., Grumberg, O., Long, D.E.:
\newblock Model checking and abstraction.
\newblock ACM Transactions on Programming Languages and Systems \textbf{16}(5)
  (1994)  1512--1542

\bibitem{Clarkson:2008:civitas}
Clarkson, M.R., Chong, S., Myers, A.C.:
\newblock Civitas: Toward a secure voting system.
\newblock In: Proc. IEEE Symposium on Security and Privacy. (May 2008)
  354--368

\bibitem{Singh:2011:Trivitas}
Bursuc, S., Grewal, G.S., Ryan, M.D.:
\newblock Trivitas: Voters directly verifying votes.
\newblock In: Proc. E-Voting and Identity. (September 2011)  190--207

\bibitem{Singh:2013:CC}
Grewal, G.S., Ryan, M.D., Bursuc, S., Ryan, P.Y.A.:
\newblock Caveat coercitor: Coercion-evidence in electronic voting.
\newblock In: Proc. IEEE Symposium on Security and Privacy. (May 2013)
  367--381

\end{thebibliography}

\end{document}